\documentclass[aps,pra,showpacs,preprintnumbers,amsmath,amssymb, superscriptaddress]{revtex4}

\usepackage{graphicx}
\usepackage{amsmath, amssymb, amsthm, amsfonts}
\usepackage{float, color}
\usepackage{ascmac, fancybox}
\usepackage{bm}

\def\R{{\mathbb R}}  
\def\C{{\mathbb C}}

\def\S{{\mathcal{S}}}
\def\M{{\mathcal{M}}}

\def\B{{\mathcal{B}}}
\def\P{{\mathcal{P}}}

\def\d{{\delta}}
\def\l{{\lambda}}
\def\m{{\mu}}

\def\r{{\rho}}
\def\s{{\sigma}}

\def\th{{\theta}}
\def\x{{\xi}}
\def\y{{\eta}}

\def\Th{{\Theta}}
\def\o{{\omega}}
\def\O{{\Omega}}

\def\Tr{{\rm Tr}\,}

\newcommand{\argmax}{\mathop{\rm arg~max}\limits}
\newcommand{\argmin}{\mathop{\rm arg~min}\limits}

\newtheorem{theorem}    {Theorem}%[section]
\newtheorem{lemma}      [theorem]{Lemma}%[section]
\newtheorem{corollary}  [theorem]{Corollary}%[section]
\begin{document}
%%%%%%%%%%%%%%%%%%%%%%%%%%%%%%%%%%%%%%%
%\title{Maximum likelihood method for quantum state tomography: \\ An information geometric approach}
\title{Data processing for qubit state tomography: \\ An information geometric approach}

\author{Akio Fujiwara}
\email{fujiwara@math.sci.osaka-u.ac.jp}
\affiliation{%
Department of Mathematics, Osaka University,
Toyonaka, Osaka 560-0043, Japan
}%

\author{Koichi Yamagata}
\affiliation{%
Department of Information and System Engineering,
Chuo University, Bunkyo-ku, Tokyo 112-8551, Japan
}%

%\author{Akio Fujiwara%
%\thanks{fujiwara@math.sci.osaka-u.ac.jp}\\
%{Department of Mathematics, Osaka University}\\ 
%{Toyonaka, Osaka 560-0043, Japan}\\ \\
%and \\ \\
%Koichi Yamagata\\
%{Department of Information and System Engineering, Chuo University} \\
%%{1-13-27, Kasuga, 
%{Bunkyo-ku, Tokyo 112-8551, Japan}
%}%
\date{\today}

%------------------------------------------------------------------------------

%\begin{document}
%\maketitle

\begin{abstract}
A statistically feasible data post-processing method for the conventional qubit state tomography is studied from an information geometrical point of view. 
%Suppose that the $i$th Pauli matrix $\s_i$ was measured $N_i$ times and obtained outcomes $+1$ (spin-up) and $-1$ (spin-down) $n^+_i$ times and $n^-_i$ times, respectively, during the course of the experiments. 
It is shown that the space $(-1,1)^3$ of the Stokes parameters $(\x_1, \x_2,\x_3)$ that specify qubit states should be regarded as a Riemannian manifold endowed with a metric $g_{ij}:=\d_{ij}/(1-(\x_i)^2)$, 
%\[ g_\x\left(\frac{\partial}{\partial\x_i}, \frac{\partial}{\partial\x_j}\right) =\frac{N_i\, \d_{ij}}{1-(\x_i)^2},\]
%where $N_i$ is the number of measurements of the $i$th Pauli observable $\s_i$ during the course of the experiments. 
%Furthermore, it the temporal estimate
%\[ \hat\x=\left(\frac{n^+_1-n^-_1}{N_1},\frac{n^+_2-n^-_2}{N_2}, \frac{n^+_3-n^-_3}{N_3}\right) \]
%for the parameter $\x$ has fallen outside the Bloch ball, then 
and that the data processing based on the maximum likelihood method is realized by the orthogonal projection from the empirical distribution onto the Bloch sphere with respect to the metric $g_{ij}$. 
An efficient algorithm for computing the maximum likelihood estimate is also proposed. 
\end{abstract}

\maketitle

%----------------------------------------------------------------------------------------------------------------------------------
\section{Introduction}
%----------------------------------------------------------------------------------------------------------------------------------

It is well known that there is a one-to-one affine correspondence between 
the quantum state space 
\[ \S(\C^2):=\{\r\; | \; \r \ge 0,\;\Tr\r=1\} \]
on the two-dimensional Hilbert space $\C^2$ 
and the unit ball 
\[ B:=\left.\left\{\x=(\x_1, \x_2, \x_3)\in\R^3 \; \right| \; \|\x\|^2:=(\x_1)^2+(\x_2)^2+(\x_3)^2 \le1 \right\} \]
in the Euclidean space $\R^3$.  
In fact, the correspondence is explicitly given by the Stokes parametrization: 
\[
% B\ni 
 \x \longmapsto
 % \r :
\r_\x
 =\frac{1}{2}(I+\x_1 \s_1+\x_2 \s_2 +\x_3 \s_3),
%\in\S(\C^2),
\]
where $\s_1$, $\s_2$, and $\s_3$ are the standard Pauli matrices. 
The unit ball $B$ in the Stokes parameter space is sometimes referred to as the Bloch ball. 
Because of the relations
\[ E_\x[\s_i]:=\Tr \r_\x \s_i=\x_i,\qquad (i\in\{1,2,3\}),  \]
the set $\s=(\s_1, \s_2, \s_3)$ of observables is regarded as an unbiased estimator \cite{{LehmanCasella}, {Helstrom:1976}, {Holevo:1982}} for the parameter $\x=(\x_1, \x_2, \x_3)$. 
This is the basic idea behind the conventional qubit state tomography. 

%To motivate our study, 
%let us consider the problem of estimating the true value of the parameters $\x$ of an unknown state $\r=\r_\x$ by the standard quantum state tomography. 
Suppose that, among $3N$ independent experiments, the $i$th Pauli matrix $\s_i$ was measured $N$ times and obtained outcomes $+1$ (spin-up) and $-1$ (spin-down), each $n_{i}^{+}$ and $n_{i}^{-}$ times.
Then a natural estimate for the true value of the parameter $\x=(\x_1,\x_2,\x_3)$ is
\begin{equation}\label{eqn:checkXi}
 \hat\x=(\hat\x_1, \hat\x_2,\hat\x_3)
 :=\left(\frac{n_{1}^{+}-n_{1}^{-}}{N},\;\frac{n_{2}^{+}-n_{2}^{-}}{N},\;\frac{n_{3}^{+}-n_{3}^{-}}{N}\right).
\end{equation}
In reality, there is a possibility that $\hat\x$ falls outside the Bloch ball $B$, because $\hat\x$ can take any value on the Stokes parameter space $[-1,1]^3$. 
In such cases, the temporal estimate $\hat\x$ must be corrected so that the new estimate falls within the Bloch ball $B$.
% from the information of the temporal estimate $\hat\x$. 
One may be tempted to adopt, as an alternative to $\hat\x$, the ``closest'' point on the Bloch sphere 
$S:=\{\x\in\R^3\,|\, \|\x\|^2=1\}$ from $\hat\x$ as measured by the Euclidean distance, i.e., the intersection of the unit sphere $S$ and the segment connecting $\hat\x$ and the origin of $\R^3$. 
Obviously, such an idea is based on Euclidean geometry, regarding the Bloch ball $B$ as a submanifold of the space $[-1,1]^3\, (\subset\R^3)$ endowed with Euclidean structure.
However, there is no a priori reason for regarding the domain $B$ of the Stokes parameters as a submanifold of  Euclidean space $\R^3$.
% endowed with Euclidean structure.

The purpose of the present paper is to clarify that such an idea for data post-processing based on Euclidean geometry is not justified from a statistical point of view, and to propose an alternative, efficient method of correcting the temporal estimate $\hat\x$ that has fallen outside the Bloch ball $B$ based on the maximum likelihood method \cite{{LehmanCasella},{Hradil:1997},{BanaszekDPS:1999},{HradilSBR:2000},{JamesKMW:2001},{deBurghLDG:2008},{BlumeKohout:2010}}. 
In what follows, we restrict ourselves to the interior $(-1,1)^3$ of the Stokes parameter space $[-1,1]^3$ to avoid statistical singularities. 
The main result of the present paper is the following

\begin{theorem}\label{thm:main1}
%From the statistical point of view, the 
In the conventional quantum state tomography, 
the Bloch ball $B$ should be regarded as a submanifold of a Riemannian manifold $(-1,1)^3$ endowed with a metric $g$ whose components at $\x\in (-1,1)^3$ are given, up to scaling, by
\begin{equation}\label{eqn:metric1}
 g_\x\left(\frac{\partial}{\partial\x_i},\;\frac{\partial}{\partial\x_j}\right)
 =\frac{\d_{ij}}{1-(\x_i)^2},\qquad (i,j\in\{1,2,3\}).
\end{equation}
If the temporal estimate $\hat\x=(\hat\x_1, \hat\x_2, \hat\x_3)\in (-1,1)^3$ has fallen outside the Bloch ball $B$,  the corrected estimate $\x^*=(\x^*_1, \x^*_2, \x^*_3)$ based on the maximum likelihood method is the orthogonal projection from $\hat\x$ onto the Bloch sphere $S$ with respect to the metric (\ref{eqn:metric1}), and is given by the unique solution of the simultaneous equations
\[
 \x^*_i \left(1-(\x^*_i)^2\right)=\l\,(\hat\x_i-\x^*_i),\qquad (i\in\{1,2,3\})
\]
and
\[  (\x^*_1)^2+(\x^*_2)^2+(\x^*_3)^2=1, \]
where $\l$ is an auxiliary positive parameter. 
\end{theorem}

It is also possible to generalize Theorem \ref{thm:main1} to treat the case when the numbers of measurements in the directions $\s_i$ are not equal. 
Suppose that, among $N$ independent experiments, the $i$th Pauli matrix $\s_i$ was measured $N_i$ times and obtained outcomes $+1$ and $-1$, each $n_{i}^{+}$ and $n_{i}^{-}$ times.
Then we have

\begin{theorem}\label{thm:main2}
%Let $N_i$ be the number of measuring the observable $\s_i$ for $i\in\{1,2,3\}$. 
%Suppose the $i$th Pauli matrix $\s_i$ is measured $N_i$ times and obtained outcomes $+1$ (spin-up) and $-1$ (spin-down) $n_{i}^{+}$ and $n_{i}^{-}$ times, respectively.
In the above-mentioned generalized quantum state tomography, 
the Bloch ball $B$ should be regarded as a submanifold of a Riemannian manifold $(-1,1)^3$ endowed with a metric $g$ whose components at $\x\in (-1,1)^3$ are given, up to scaling, by
\begin{equation}\label{eqn:metric2}
 g_\x\left(\frac{\partial}{\partial\x_i},\;\frac{\partial}{\partial\x_j}\right)
 =\frac{\hat s_i\,\d_{ij}}{1-(\x_i)^2},\qquad (i,j\in\{1,2,3\}),
\end{equation}
where $\hat s_i:=N_i/N$. 
If the temporal estimate 
\[ 
 \hat\x=(\hat\x_1, \hat\x_2,\hat\x_3)
 :=\left(\frac{n_{1}^{+}-n_{1}^{-}}{N_1},\;\frac{n_{2}^{+}-n_{2}^{-}}{N_2},\;\frac{n_{3}^{+}-n_{3}^{-}}{N_3}\right)
\]
has fallen outside the Bloch ball $B$, the corrected estimate $\x^*=(\x^*_1, \x^*_2, \x^*_3)$ based on the maximum likelihood method is the orthogonal projection from $\hat\x$ onto the Bloch sphere $S$ with respect to the metric (\ref{eqn:metric2}), and is given by the unique solution of the simultaneous equations
\[
\x^*_i \left(1-(\x^*_i)^2\right)=\l \hat s_i\,(\hat\x_i-\x^*_i),\qquad (i\in\{1,2,3\})
\]
and
\[  (\x^*_1)^2+(\x^*_2)^2+(\x^*_3)^2=1, \]
where $\l$ is an auxiliary positive parameter. 
\end{theorem}

The paper is organized  as follows.
In Section \ref{sec:2}, we first review the maximum likelihood method from a geometrical point of view, and then prove Theorem~\ref{thm:main1} by establishing an isomorphism between the Stokes parameter space and the statistical manifold of independent probability distributions. 
In Section \ref{sec:3}, we introduce the notion of randomized tomography, and prove Theorem~\ref{thm:main2} by analyzing the statistical nature of randomized tomography using the technique of mutually orthogonal dualistic foliations. 
In section \ref{sec:4}, we devise an efficient algorithm for computing the maximum likelihood estimate $\x^*$.
% in Theorems~\ref{thm:main1} and \ref{thm:main2} . 
Section \ref{sec:5} is devoted to conclusions. 
Throughout the paper, we make use of some basic knowledge of information geometry \cite{{AmariNagaoka},{AmariLN},{MurrayRice}}, 
and therefore, we give a brief overview of information geometry in Appendix for the reader's convenience.
%For the reader's convenience, some basic materials are provided in Appendices, including operator monotone functions, information geometry, and G\^ateaux derivatives. 

%----------------------------------------------------------------------------------------------------------------------------------
\section{Proof of Theorem \ref{thm:main1}}\label{sec:2}
%----------------------------------------------------------------------------------------------------------------------------------

%-------------------------------------------------------------------------------------------
\subsection{Maximum likelihood method}
%-------------------------------------------------------------------------------------------

Let $\P(\O)$ denote the set of probability distributions %(i.e., classical states) 
on a finite sample space $\O$, i.e.,
\[
 \P(\O):=\left\{p:\O\to\R \; \left| \; p(\o)>0 \mbox{  for all $\o\in\O$, and } \sum_{\o\in\O}p(\o)=1 \right\}\right..
\]
This set may be identified with the $(|\O|-1)$-dimensional (open) simplex, where $|\O|$ denotes the number of elements in $\O$, 
and thus it is sometimes referred to as the {\em probability simplex} on $\O$.  
%According to information geometry \cite{AmariNagaoka}, 
%Recall that $\P(\O)$ 
The set $\P(\O)$ is also regarded as a {\em statistical manifold} endowed with the {\em dualistic structure} $(g, \nabla^{(e)}, \nabla^{(m)})$, where $g$ is the {\em Fisher metric}, and $\nabla^{(e)}$ and $\nabla^{(m)}$ are the {\em exponential} and {\em mixture connections}, (cf., Appendix). 

Suppose that the state of the physical system at hand belongs to a (closed) subset $\M$ of $\P(\O)$, but we do not know which is the true state. 
We further assume that the probability distributions of $\M$ are faithfully parametrized by a finite dimensional parameter $\th$ as
\[ \M=\{p_\th(\o)\;|\; \th\in\Th\}. \] 
In this case, $\M$ is called a {\em parametric model}, 
and our task is to estimate the true value of the parameter $\th$ that specifies the true state. 
%\[ \P=\{p_\th(\o)\;|\; \o\in\O,\; \th\in\Th\subset\R^d\}, \] 
%that is, a parametric family of probability density functions $p_\th(\o)$ on a finite sample space $\O$ having a $d$-dimensional parameter $\th\in\Th$. 
Suppose that, by $n$ independent experiments, we have obtained data $(x_1,x_2,\dots,x_n)\in\O^n$. 
This information is 
%represented by 
compressed into the {\em empirical distribution}, an element of $\P(\O)$ defined by
\begin{eqnarray*}
  \hat q_n(\o)
 &:=&\frac{\mbox{Number of occurrences of $\o$ in data $(x_1,x_2,\dots,x_n)$}}{n}\\
 &=&\frac{1}{n}\sum_{i=1}^n \d_{x_i}(\o)
\end{eqnarray*}
for each $\o\in\O$, where $\d_{x_i}(\o)$ is the Kronecker delta.
% that takes the value $1$ if $\o=\o_i$ and $0$ otherwise. 
If $\hat q_n$ belongs to the model $\M$, then we have an estimate $\hat\th_n$ that satisfies $p_{\hat\th_n}=\hat q_n$. 
However, the empirical distribution $\hat q_n$ does not always belong to the model $\M$. 
When $\hat q_n\notin \M$, we need to find an alternative estimate from the data.
One of the standard method of finding an alternative estimate $p_{\hat\th_n}\in\M$ is the {\em maximum likelihood method}, 
in which one seeks the maximizer of the likelihood function 
\[ \th\longmapsto p_{\th}(x_1)p_{\th}(x_2)\dots p_{\th}(x_n), \]
%given the data $(\o_1,\o_2,\dots,\o_n)$:
within the domain $\Th$ of the parameter $\th$, so that
\begin{equation}\label{eqn:naiveMLE}
 \hat \th_n:=\argmax_{\th\in\Th} \left\{p_{\th}(x_1)p_{\th}(x_2)\cdots p_{\th}(x_n)\right\}.
\end{equation}
We can rewrite this relation as follows.
\begin{eqnarray*}
 \hat \th_n
 &=&\argmax_{\th\in\Th} \;\frac{1}{n}\sum_{i=1}^n \log p_{\th}(x_i) \\
 &=&\argmax_{\th\in\Th} \;\sum_{\o\in\O} \hat q_n(\o) \log p_{\th}(\o) \\
 &=&\argmin_{\th\in\Th} \;\sum_{\o\in\O} \hat q_n(\o) \left\{\log\hat q_n(\o)-\log p_{\th}(\o)\right\} \\
 &=&\argmin_{\th\in\Th} \; D(\hat q_n\|p_\th),
\end{eqnarray*}
where 
\[ D(q\|p):=\sum_{\o\in\O} q(\o) \log\frac{q(\o)}{p(\o)} \]
is the {\em Kullback-Leibler divergence} from $q$ to $p$. 
%or the {\em relative entropy}. 
In other words, the maximum likelihood estimate \cite{note}
%\footnote{
%In the present paper, we use the term ``maximum likelihood estimate'' for both the parameter $\hat\th_n$ and the corresponding probability distribution $p_{\hat\th_n}$.}
(MLE) $p_{\hat\th_n}$ is the point on $\M$ that is ``closest'' from the empirical distribution $\hat q_n$ as measured by the Kullback-Leibler divergence:
\begin{equation}\label{eqn:MLE1}
p_{\hat\th_n}=\argmin_{p\in\M} \; D(\hat q_n\|p).
\end{equation}
Due to the generalized Pythagorean theorem (cf., Appendix), the MLE is geometrically understood as the $\nabla^{(m)}$-projection from $\hat q_n$ to $\M$ or its boundary, as illustrated in Fig.~\ref{fig:mProj}.
%where $\nabla^{(m)}$ is the mixture connection of $\P(\O)$  
%It is unique if $\M$ is $\nabla^{(m)}$-convex. 

\begin{figure}[t] %---------------------------------------------------------------------------------
	\begin{centering}
	\includegraphics[scale=0.6]{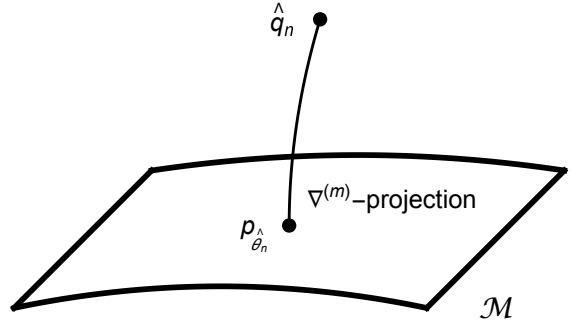}
	\par
	\end{centering}
	\caption{The maximum likelihood estimate $p_{\hat\th_n}$ is the minimizer of the function 
		$p\mapsto D(\hat q_n\| p)$ with respect to $p\in\M$, and is also understood as the
		$\nabla^{(m)}$-projection from the empirical distribution $\hat{q}_n$ to $\M$ or its boundary. 
	\label{fig:mProj}}
\end{figure}%---------------------------------------------------------------------------------

%-------------------------------------------------------------------------------------------
\subsection{Manifold of product distributions}
%\subsection{Geometry of $\P(\O)^{\otimes k}$ as a submanifold of $\P(\O^k)$}
%-------------------------------------------------------------------------------------------

Let us consider, for each $i=1,\dots,k$, a coin flipping model 
\[
 p_{\x_i}(\o_i)=\left\{\array{ll} 
 \displaystyle \frac{1+\x_i}{2},& (\o_i=+1) \\ \\
 \displaystyle \frac{1-\x_i}{2}, & (\o_i=-1) \endarray\right.
\]
on $\O=\{-1,+1\}$ having a one-dimensional parameter $\x_i\in (-1,1)$, 
and let us denote their product distribution by
\begin{equation}\label{eqn:independentMfd}
 p_{\x}(\o):=\prod_{i=1}^k p_{\x_i}(\o_i),
\end{equation}
where $\x:=(\x_1,\dots,\x_k)\in (-1,1)^k$ and $\o=(\o_1,\dots,\o_k)\in \O^k$. 
%We denote the totality of those product distributions by $\P(\O)^{\otimes k}$.
%The set $\P(\O)^{\otimes k}$, 
The set
\[ 
 \P(\O)^{\otimes k}:=\left\{p_\x(\o)\in\P(\O^k)\; \left| \; \x\in (-1,1)^k \right.\right\}, 
\]
comprising independent probability distributions, 
is regarded as a $k$-dimensional submanifold, having a (global) coordinate system $\x$, 
embedded in the $(2^k-1)$-dimensional statistical manifold $\P(\O^k)$.

The submanifold $\P(\O)^{\otimes k}$ is not $\nabla^{(m)}$-autoparallel (i.e., not a mixture family) unless $k=1$, 
but it is $\nabla^{(e)}$-autoparallel (i.e., an exponential family) because (\ref{eqn:independentMfd}) is rewritten as
\begin{eqnarray*}
 p_\x(\o)
 &=&\prod_{i=1}^k \exp\left[ \log\frac{p_{\x_i}(+1)}{p_{\x_i}(-1)}\,\d_{+1}(\o_i)+\log p_{\x_i}(-1)  \right] \\
% &=&\exp\left[ \sum_{i=1}^k \log\frac{p_{\x_i}(+1)}{p_{\x_i}(-1)}\,\d_{+1}(\o_i)+\log p_{\x_i}(-1)  \right] \\
 &=&\exp\left[ \left(\sum_{i=1}^k \log\frac{1+\x_i}{1-\x_i}\,\d_{+1}(\o_i)\right)+\left(\sum_{i=1}^k\log p_{\x_i}(-1)\right)  \right] \\
 &=&\exp\left[ \left(\sum_{i=1}^k \th^i F_i(\o)\right)-\overline\psi(\overline\th)  \right],
\end{eqnarray*}
where $F_i(\o):=\d_{+1}(\o_i)$, 
\begin{equation}\label{eqn:theta}
 \th^i:=\log\frac{1+\x_i}{1-\x_i},\qquad (i\in\{1,\dots,k\}), 
\end{equation}
and
\begin{equation}\label{eqn:psi}
 \overline\psi(\overline\th):=-\sum_{i=1}^k \log p_{\x_i}(-1)=\sum_{i=1}^k \log\left(1+e^{\th^i}\right)
\end{equation}
with $\overline\th:=(\th^1,\dots,\th^k)$. 
The parameters $\overline\th=(\th^1,\dots,\th^k)$ form a $\nabla^{(e)}$-affine coordinate system of $\P(\O)^{\otimes k}$, 
and its dual %($\nabla^{(m)}$-affine) 
coordinate system $\overline\y=(\y_1,\dots,\y_k)$ is given by 
\begin{equation}\label{eqn:eta}
 \y_i:=\frac{\partial\overline\psi}{\partial\th^i}=\frac{e^{\th^i}}{1+e^{\th^i}}=\frac{1+\x_i}{2},\qquad (i\in\{1,\dots,k\}).
\end{equation}
%Note that these parameters are identical to the expectation parameters $\y_i=E_{p_\x}[F_i]$.

Now let us return to the quantum state tomography. 
The conventional quantum state tomography is regarded as $N$-round experiments, each round being composed of three independent measurements of observables $\s_1, \s_2$, and $\s_3$. 
Mathematically, each round of the experinemt is isomorphic to the case $k=3$ in the above coin flipping model, with $\x=(\x_1,\x_2,\x_3)$ being the Stokes parameters. 
The condition   
\begin{equation}\label{eqn:BlochBall}
 \|\x\|^2=(\x_1)^2+(\x_2)^2+(\x_3)^2 \le1
\end{equation}
%This condition 
defines a subset $\B$ of $\P(\O)^{\otimes 3}$ through the parametrization (\ref{eqn:independentMfd}).
% that corresponds to the quantum state space $\S(\C^2)$.
% to the probability distributions that correspond to quantum states in $\S(\C^2)$ form a subset $\B$ of $\P(\O)^{\otimes 3}$.
Given a temporal estimate $\hat\x=(\hat\x_1,\hat\x_2,\hat\x_3)$ for the Stokes parameters $\x$ through (\ref{eqn:checkXi}), 
let the corresponding product distribution be 
\[
 \hat q_{3N}(\o_1,\o_2,\o_3)
% =\prod_{i=1}^3 \hat q_N(\o_i)
 :=\prod_{i=1}^3 p_{\hat\x_i}(\o_i),
% =p_{\hat\x}(\o_1,\o_2,\o_3),
\]
which is regarded as the empirical distribution for the quantum state tomography. 
%that corresponds to the temporal estimate $\hat\x=(\hat\x_1,\hat\x_2,\hat\x_3)$ being given by (\ref{eqn:checkXi}), 
Although the distribution $\hat q_{3N}$ belongs to $\P(\O)^{\otimes 3}$, it does not always belong to $\B$. 
Thus, in order to obtain a physically valid estimate that belongs to $\B$, we may apply the maximum likelihood method, to obtain   
\begin{equation}\label{eqn:MLE2}
p^*
%:=p_{\hat\x^*}
:=\argmin_{p\in\B} \; D(\hat q_{3N}\|p).
\end{equation}
As mentioned in the previous subsection, this amounts to finding the $\nabla^{(m)}$-projection from $\hat q_{3N}$ to $\B$ or its boundary.
Although both $\hat q_{3N}$ and $p^*$ belong to $\P(\O)^{\otimes 3}$, 
the $\nabla^{(m)}$-geodesic connecting $\hat q_{3N}$ and $p^*$ in $\P(\O^3)$ 
%runs outside 
does not stay within 
$\P(\O)^{\otimes 3}$ because $\P(\O)^{\otimes 3}$ is not $\nabla^{(m)}$-autoparallel in $\P(\O^3)$.
Consequently, the $\nabla^{(m)}$-projection from the empirical distribution $\hat q_{3N}$ to $\B$ 
%\mapsto p^*$ 
in $\P(\O^3)$ cannot be immediately interpreted as a certain projection from the temporal estimate $\hat\x$ to the Bloch ball $B$ in the Stokes parameter space $(-1,1)^3$.
%s $\x=(\x_1,\x_2,\x_3)$ since their domain $(-1,1)^3$ corresponds to $\P(\O)^{\otimes 3}$. 

\begin{figure}[t] %---------------------------------------------------------------------------------
	\begin{centering}
	\includegraphics[scale=0.5]{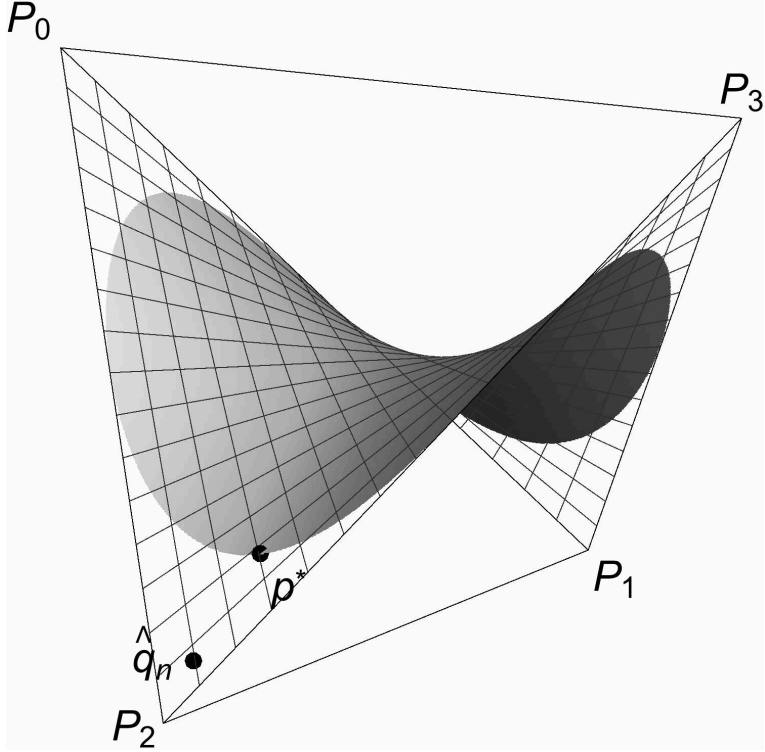}
	\par
	\end{centering}
	\caption{Geometry of two-dimensional quantum state tomography. 
		The set $\B$ of physically valid states (deformed grayish disk) is a subset of 
		the two-dimensional manifold $\P(\O)^{\otimes 2}$ of independent distributions (ruled surface) 
		embedded in the three-dimensional probability simplex $\P(\O^2)$ 
		(convex hull of $\{P_0, P_1, P_2, P_3\}$). 
		The maximum likelihood estimator $p^*$ is the point in $\B$
		that is ``closest'' from the empirical distribution $\hat q_n$ 
		as measured by the Kullback-Leibler divergence. 
%		the $\nabla^{(m)}$-projection onto $\M$ from the empirical distribution $\hat{q}_n$. 
	\label{fig:indMfd}}
\end{figure}%---------------------------------------------------------------------------------

In order to get a better understanding of the above-mentioned difficulty, let us consider the case when $k=2$: 
this situation may be interpreted as the quantum state tomography restricted to the $\x_1\x_2$-plane. 
%the geometrical configuration of which is illustrated in Fig.~\ref{fig:indMfd}.
%Note that this situation corresponds to the quantum state tomography restricted to the $\x_1\x_2$-plane.  
Fig.~\ref{fig:indMfd} depicts the relationship between $\P(\O^2)$ and $\P(\O)^{\otimes 2}$, as well as the subset $\B$ that corresponds to the quantum state space $\S(\C^2)$. 
The statistical manifold $\P(\O^2)$ is a 3-dimensional simplex represented by the convex hull of four points $P_0$, $P_1$, $P_2$, and $P_3$, each corresponding to the $\d$-measure on the events $(+1,+1)$, $(+1,-1)$, $(-1,+1)$, and $(-1,-1)$, respectively.
The ruled surface embedded in the simplex corresponds to the submanifold $\P(\O)^{\otimes 2}$ of independent  distributions, 
and the deformed grayish disk lying on the ruled surface represents the subset $\B$ of physically valid states satisfying $(\x_1)^2+(\x_2)^2\le 1$. 
%The submanifold $\P(\O)^{\otimes 2}$ of independent probability distributions is the ruled surface embedded in the simplex, 
%and the subset $\M$ of physically valid states satisfying $(\x_1)^2+(\x_2)^2\le 1$ is the deformed grayish disk that lay on the ruled surface. 
%that corresponds to the $\x_1\x_2$-plane tomography, 
%that satisfy $(\x_1)^2+(\x_2)^2\le 1$ is depicted as the deformed grayish disk. 
Now suppose that the empirical distribution $\hat q_n\in \P(\O)^{\otimes 2}$ has fallen outside $\B$.
The MLE $p^*$ is then given by the point on $\B$ that is ``closest'' from $\hat q_n$ as measured by the Kullback-Leibler divergence.
Since the ruled surface $\P(\O)^{\otimes 2}$ is embedded in the simplex $\P(\O^2)$ as a ``curved'' surface, the $\nabla^{(m)}$-geodesic (straight line) connecting $\hat q_n$ and $p^*$ in $\P(\O^2)$ does not stay within $\P(\O)^{\otimes 2}$.
Recall that there is a one-to-one correspondence between the set $\P(\O)^{\otimes 2}$ of independent distributions and the Stokes parameter space $(-1,1)^2$.
Thus, the $\nabla^{(m)}$-geodesic connecting $\hat q_n$ and $p^*$ in $\P(\O^2)$ has no direct counterpart in the Stokes parameter space.
% that corresponds to the connecting $\hat q_n$ and $p^*$ in the simplex $\P(\O^2)$. 
%, although both ends $\hat q_n$ and $p^*$ belong to $\P(\O)^{\otimes 2}$. 

This difficulty can be surmounted by introducing a dualistic structure $(\overline g, \overline\nabla^{(e)}, \overline\nabla^{(m)})$ on the submanifold $\P(\O)^{\otimes k}$ as the restriction of the dualistic structure $(g, \nabla^{(e)}, \nabla^{(m)})$ of the ambient statistical manifold $\P(\O^k)$ onto $P(\O)^{\otimes k}$.
%\cite{AmariNagaoka}. 
Since $\P(\O)^{\otimes k}$ is a $\nabla^{(e)}$-autoparallel submanifold of $\P(\O^k)$,   $\P(\O)^{\otimes k}$ is automatically dually flat with respect to the induced structure $(\overline g, \overline\nabla^{(e)}, \overline\nabla^{(m)})$, and the parameters $\overline\th$ and $\overline\y$ defined by (\ref{eqn:theta}) and (\ref{eqn:eta}) form mutually dual $\overline\nabla^{(e)}$- and $\overline\nabla^{(m)}$-affine coordinate systems of $\P(\O)^{\otimes k}$. 
Let us denote the canonical $\overline\nabla^{(m)}$-divergence on $\P(\O)^{\otimes k}$ by $\overline D(p\|q)$. 
Note that the canonical $\nabla^{(m)}$-divergence on the ambient manifold $\P(\O^k)$ is nothing but the Kullback-Leibler divergence $D(p\|q)$. 
The key observation is the following 
%\cite{FujiwaraAmari:1995}. 

\begin{lemma}\label{lem:divergenceEquality} %{\rm (\cite{FujiwaraAmari:1995})}
For any $p,q\in \P(\O)^{\otimes k}$, we have
\[ \overline D(p\|q)=D(p\|q). \]
\end{lemma}

\begin{proof} 
The assertion has been proved under a more general setting in \cite{FujiwaraAmari:1995}; however, we shall give an alternative proof for the sake of later discussion. 
%for the reader's convenience. 
%considering a dualistic orthogonal foliation \cite{AmariNagaoka} of $\P(\O^k)$ based on the $\nabla^{(e)}$-autoparallel submanifold $\P(\O)^{\otimes k}$.
%, a standard technique in information geomety.
%Let $d:=2^k-1$ and 
Let $\overline\th=(\th^1,\dots,\th^k)$ and $\overline\y=(\y_1,\dots,\y_k)$ be mutually dual affine coordinate systems of $\P(\O)^{\otimes k}$ defined by (\ref{eqn:theta}) and  (\ref{eqn:eta}), respectively.
By extending these coordinate systems, we construct mutually dual $\nabla^{(e)}$- and $\nabla^{(m)}$-affine coordinate systems 
%of $\P(\O^k)$ as
\begin{equation}\label{eqn:lem3theta}
 \th=(\th^1,\dots,\th^k;\,\th^{k+1},\dots,\th^d)
\end{equation}
and
% \quad\mbox{and}\quad 
\begin{equation}\label{eqn:lem3eta}
 \y=(\y_1,\dots,\y_k;\,\y_{k+1},\dots,\y_d) 
\end{equation}
%be mutually dual $\nabla^{(e)}$- and $\nabla^{(m)}$-affine coordinate systems 
of $\P(\O^k)$, with $d:=2^k-1$, 
such that the $\nabla^{(e)}$-autoparallel submanifold $\P(\O)^{\otimes k}$ corresponds to the points satisfying 
%is represented as
\begin{equation}\label{eqn:lem3submanifold}
 (\th^{k+1},\dots,\th^d)=(0,\dots,0).
\end{equation} 
%This is possible because $\P(\O)^{\otimes k}$ is a $\nabla^{(e)}$-autoparallel submanifold of $\P(\O^k)$. 
Furthermore, let $\psi(\th)$ and $\varphi(\y)$ be the dual potentials for the dual affine coordinate systems $\th$ and $\y$ of $\P(\O^k)$ satisfying
\begin{equation}\label{eqn:lem3potentials}
 \psi(\th)+\varphi(\y)-\th \cdot \y=0,
\end{equation}
where $\cdot$ denotes the standard inner product, and
\begin{equation}\label{eqn:lem3OverlinePsi}
 \psi(\overline\th;\,0,\dots,0)= \overline\psi(\overline\th),
\end{equation}
where $\overline\psi(\overline\th)$ is the potential function on $\P(\O)^{\otimes k}$ defined by 
(\ref{eqn:psi}).
Note that the dual potential function $\overline\varphi(\overline\y)$ on $\P(\O)^{\otimes k}$ is defined by
\begin{equation}\label{eqn:lem3OverlineVarphi}
 \overline\varphi(\overline\y):=\overline\th\cdot\overline\y-\overline\psi(\overline\th). 
\end{equation}
%be the dual potential function of $\overline\psi(\overline\th)$ on $\P(\O)^{\otimes k}$. 
%Recall that
Now, since the Kullback-Leibler divergence $D(p\|q)$ is the $\nabla^{(m)}\,(=\nabla^{(e)*})$-divergence, 
%i.e., the $\nabla^{(e)*}$-divergence on $\P(\O^k)$, 
we have
\begin{eqnarray*}
 D(p\|q)
%&=& D^{\nabla^{(e)*}}(p\|q) \\
&=& \psi(\th(q))+\varphi(\y(p))-\th(q) \cdot \y(p) \\
&=& \psi(\th(q))+\{\th(p) \cdot \y(p)-\psi(\th(p))\}
	-\th(q) \cdot \y(p) \\
&=& \psi(\th(q))-\psi(\th(p))
	+\{\th(p)-\th(q)\} \cdot \y(p)
\end{eqnarray*}
where $\th(q)$, for instance, stands for the $\th$-coordinate of the point $q\in\P(\O^k)$, 
and the identity (\ref{eqn:lem3potentials}) was used in the second equality. 
Furthermore, since both $p$ and $q$ belong to the submanifold $\P(\O)^{\otimes k}$,
% on which (\ref{eqn:lem3submanifold}) holds, 
we have
\begin{eqnarray*}
 D(p\|q)
&=& \overline\psi(\overline\th(q))-\overline\psi(\overline\th(p))+(\overline\th(p)-\overline\th(q);\,0,\dots,0) \cdot \y(p) \\
&=& \overline\psi(\overline\th(q))-\overline\psi(\overline\th(p))+\{\overline\th(p)-\overline\th(q)\} \cdot \overline\y(p) \\
&=& \overline\psi(\overline\th(q))-\{\overline\th(p) \cdot \overline\y(p)-\overline\varphi(\overline\y(p))\}+\{\overline\th(p)-\overline\th(q)\} \cdot \overline\y(p) \\
&=& \overline\psi(\overline\th(q))+\overline\varphi(\overline\y(p))-\overline\th(q) \cdot \overline\y(p) \\
&=& \overline D(p\|q).
\end{eqnarray*}
Here, the first equality is due to (\ref{eqn:lem3submanifold}) and (\ref{eqn:lem3OverlinePsi}), and the third due to (\ref{eqn:lem3OverlineVarphi}). 
This proves the claim.
\end{proof}

It follows from Lemma \ref{lem:divergenceEquality} that the MLE (\ref{eqn:MLE2}) can be rewritten as
\begin{equation}\label{eqn:MLE3}
p^*=\argmin_{p\in\B} \; \overline D(\hat q_{3N}\|p).
\end{equation}
This relation allows us to interpret the MLE $p^*$ in terms of the intrinsic geometry of the manifold $\P(\O)^{\otimes 3}$, without reference to the ambient manifold $\P(\O^3)$. 
To be specific, the MLE $p^*$ is the $\overline\nabla^{(m)}$-projection from $\hat q_{3N}$ to $\B$ in $\P(\O)^{\otimes 3}$, and the $\overline\nabla^{(m)}$-geodesic connecting $\hat q_{3N}$ and $p^*$ stays (of course!) within $\P(\O)^{\otimes 3}$. 

%-------------------------------------------------------------------------------------------
\subsection{Relation between $\P(\O)^{\otimes k}$ and $(-1,1)^k$}
%-------------------------------------------------------------------------------------------

In the previous subsection, we interpreted the projection $\hat q_{3N}\mapsto p^*$ using an intrinsic geometry of $\P(\O)^{\otimes 3}$. 
In this subsection, we further interpret the process of finding the MLE 
%from the temporal estimate $\hat\x$ 
using an intrinsic geometry of the Stokes parameter space $(-1,1)^3$. 

%We first observe that $\P(\O)^{\otimes k}$ is diffeomorphic to $(-1,1)^k$. 
%In fact, 
Firstly, we recall that the coordinate system $\overline\y=(\y_i)$ of $\P(\O)^{\otimes k}$ and the coordinate system $\x=(\x_i)$ of $(-1,1)^k$ are related by (\ref{eqn:eta}), i.e., 
\[
 \y_i=\frac{1+\x_i}{2}.
\]
This correspondence establishes a diffeomorphism $f:(-1,1)^k \to \P(\O)^{\otimes k}$.
Secondly, we introduce a Riemannian metric $\tilde{g}$ on $(-1,1)^k$ by
\[
 \tilde{g}_{p}(X,Y):=\overline g_{f(p)}(f_*X,f_*Y),\qquad \left(p\in (-1,1)^k \right)
\]
%where $X$ and $Y$ are vector fields on $(-1,1)^k$, 
where $\overline g$ is the Fisher metric on $\P(\O)^{\otimes k}$, and $f_*$ is the differential map of $f$. 
Thirdly, we introduce an affine connection $\tilde{\nabla}^{(m)}$ on $(-1,1)^k$ such that the coordinate system $\x=(\x_i)$ becomes $\tilde{\nabla}^{(m)}$-affine. 
This is nothing but the Euclidean connection induced from the natural affine structure of the ambient space $\R^k$. 
Finally, we introduce another affine connection $\tilde{\nabla}^{(e)}$ on $(-1,1)^k$ such that it satisfies the duality
\[
 X\tilde{g}(Y,Z)=\tilde{g}(\tilde{\nabla}^{(e)}_XY,Z)+\tilde{g}(Y, \tilde{\nabla}^{(m)}_XZ).
\]
In this way, we can regard the space $(-1,1)^k$ as a dually flat statistical manifold endowed with the dualistic structure $(\tilde{g}, \tilde{\nabla}^{(e)}, \tilde{\nabla}^{(m)})$.

Let us calculate the metric $\tilde g$ explicitly. 
From the relation (\ref{eqn:independentMfd}), we have
\[
 \frac{\partial}{\partial\x_i} \log p_\x(\o)
= \frac{\partial}{\partial\x_i} \log p_{\x_i}(\o_i)
=\left\{\array{ll} 
 \displaystyle \frac{1}{1+\x_i},& (\o_i=+1) \\ \\
 \displaystyle \frac{-1}{1-\x_i}, & (\o_i=-1) \endarray\right..
\]
Consequently,
\begin{eqnarray*}
\tilde g_{p_\x}\left(\frac{\partial}{\partial\x_i}, \frac{\partial}{\partial\x_i}\right)
&=&
\sum_{\o\in\O^k} p_\x(\o) \left(\frac{\partial}{\partial\x_i} \log p_\x(\o)\right)^2  \\
&=&
\frac{1+\x_i}{2} \left(\frac{1}{1+\x_i}\right)^2+ \frac{1-\x_i}{2} \left(\frac{-1}{1-\x_i}\right)^2 \\
&=&
\frac{1}{1-(\x_i)^2},
\end{eqnarray*}
and for $i\neq j$, 
\begin{eqnarray*}
\tilde g_{p_\x}\left(\frac{\partial}{\partial\x_i}, \frac{\partial}{\partial\x_j}\right)
&=&
\sum_{\o\in\O^k} p_\x(\o) \left(\frac{\partial}{\partial\x_i} \log p_\x(\o)\right) 
		\left(\frac{\partial}{\partial\x_j} \log p_\x(\o)\right) \\
%&=&
%\sum_{(\o_i,\o_j)\in\O^2}
%\frac{1+\o_i\x_i}{2}\frac{1+\o_j\x_j}{2}\left(\frac{\o_i}{1+\o_i\x_i}\right)\left(\frac{\o_j}{1+\o_j\x_j}\right)\\
&=&
\left[\sum_{\o_i\in\O} p_{\x_i}(\o_i) \left(\frac{\partial}{\partial\x_i} \log p_{\x_i}(\o_i)\right)\right]
\left[\sum_{\o_j\in\O} p_{\x_j}(\o_j) \left(\frac{\partial}{\partial\x_j} \log p_{\x_j}(\o_j)\right)\right] \\
&=&
\left[\sum_{\o_i\in\O} \left(\frac{\partial}{\partial\x_i} p_{\x_i}(\o_i)\right)\right]
\left[\sum_{\o_j\in\O} \left(\frac{\partial}{\partial\x_j} p_{\x_j}(\o_j)\right)\right] \\ \\
&=&0.
\end{eqnarray*}
In summary,
\begin{equation}\label{eqn:FisherStokes}
\tilde g_{p_\x}\left(\frac{\partial}{\partial\x_i}, \frac{\partial}{\partial\x_j}\right)
=\frac{\d_{ij}}{1-(\x_i)^2}.
\end{equation}
When $k=3$, this is identical to (\ref{eqn:metric1}). 

Now let us proceed to investigating the relationship between $(-1,1)^k$ and $\P(\O)^{\otimes k}$. 
We say two statistical manifolds $(\tilde{M}, \tilde{g}, \tilde{\nabla}, \tilde{\nabla}^*)$ and $(\overline M, \overline g, \overline\nabla, \overline\nabla^*)$ are {\em statistically isomorphic}, or simply {\em isostatistic}, if there is a diffeomorphism $f:\tilde{M}\to \overline M$ such that 
\[
 \tilde{g}_p(X,Y)=\overline g_{f(p)}(f_*X,f_*Y),\quad
 f_*(\tilde{\nabla}_XY)_p=(\overline\nabla_{f_*X} f_*Y)_{f(p)},\quad
 f_*(\tilde{\nabla}^*_XY)_p=(\overline\nabla^*_{f_*X} f_*Y)_{f(p)}
\]
holds for all $p\in \tilde M$ and vector fields $X, Y$ on $\tilde M$.

\begin{lemma}\label{lem:isostatistic}
The manifolds $((-1,1)^k, \tilde{g}, \tilde{\nabla}^{(e)}, \tilde{\nabla}^{(m)})$ and $(\P(\O)^{\otimes k}, \overline g, \overline\nabla^{(e)}, \overline\nabla^{(m)})$ are isostatistic.
\end{lemma}

\begin{proof}
Let $f:(-1,1)^k \to \P(\O)^{\otimes k}$ be the diffeomorphism defined above. 
Then 
\[ \tilde{g}_p(X,Y)=\overline g_{f(p)}(f_*X,f_*Y) \]
is obvious from the definition. 
Since $\x$ is a $\tilde{\nabla}^{(m)}$-affine coordinate system of $(-1,1)^k$ and $\y$ is a $\overline\nabla^{(m)}$-affine coordinate system of $\P(\O)^{\otimes k}$,  
\[ f_*(\tilde{\nabla}^{(m)}_{\partial^i} \partial^j)_p=0=(\overline\nabla^{(m)}_{f_*\partial^i} f_*\partial^j)_{f(p)} \]
for all $i,j\in\{1,\dots,k\}$, where $\partial^i:=\partial/\partial\x_i$. 
Finally, since $f$ is a diffeomorphism,
%-----------------------------------------------------------------
%\footnote{\color{red}
%NB. Let $\tilde G(p):=\tilde{g}_p(X,Y)$ and let $\overline G(f(p)):=\overline g_{f(p)}(f_*X,f_*Y)$. 
%Then we see from the definition that $\tilde G(p)=\overline G(f(p))=(f^*\overline G)(p)$, 
%that is, $f^*\overline G=\tilde G$. 
%Thus, due to a general formula $X_p(f^*\overline G)=(f_*X)_{f(p)}\overline G$, which holds 
%for any diffeomorphism $f$, we have
%$X_p\,\tilde{g}(Y, Z)=(f_*X)_{f(p)} \overline g(f_*Y, f_*Z)$.
%}, 
%----------------------------------------------------------------------------------------------------------------
\begin{eqnarray*}
\tilde{g}_p(\tilde{\nabla}^{(e)}_XY, Z)
&=&
X_p\,\tilde{g}(Y, Z)-\tilde{g}_p(Y, \tilde{\nabla}^{(m)}_X Z) \\
&=&
(f_*X)_{f(p)} \overline g(f_*Y, f_*Z)-\overline g_{f(p)}(f_*Y, \overline\nabla^{(m)}_{f_*X} f_*Z) \\
&=&
\overline g_{f(p)}(\overline\nabla^{(e)}_{f_*X} f_*Y, f_*Z), 
\end{eqnarray*}
which leads us to
\[
f_*(\tilde{\nabla}^{(e)}_XY)_p=(\overline\nabla^{(e)}_{f_*X} f_*Y)_{f(p)}.
\]
This proves the assertion. 
\end{proof}

Returning to the quantum state tomography, Lemma \ref{lem:isostatistic} implies that the Stokes parameter space $(-1,1)^3$ endowed with the dualistic structure $(\tilde{g}, \tilde{\nabla}^{(e)}, \tilde{\nabla}^{(m)})$ 
% introduced above, 
can be identified with the statistical manifold $\P(\O)^{\otimes 3}$ of product distributions. 
Combining this fact with the results in the previous subsection, we have the following

\begin{corollary}\label{cor:StokesMLE}
The MLE $\x^*$ that satisfies $p^*=p_{\x^*}$ 
%as (\ref{eqn:MLE2}) 
is the $\tilde{\nabla}^{(m)}$-projection from the temporal estimate $\hat\x$ to the Bloch ball $B$ in the Stokes parameter space $(-1,1)^3$.
\end{corollary}

\begin{figure}[t] %---------------------------------------------------------------------------------
	\begin{centering}
	\includegraphics[scale=0.5]{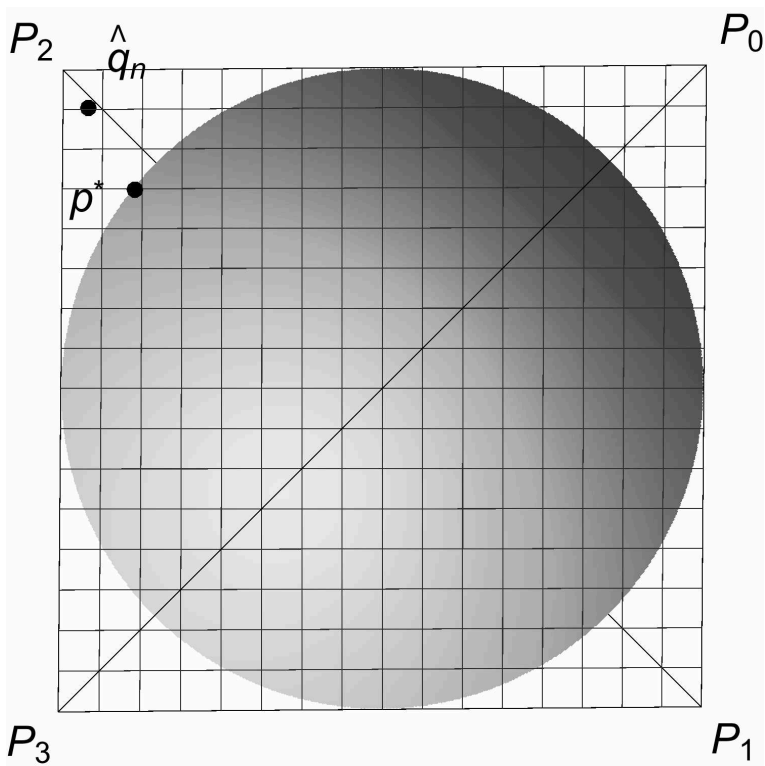}
	\par
	\end{centering}
	\caption{Geometry of two-dimensional quantum state tomography 
		as seen from the top of Fig.~\ref{fig:indMfd}. 
		This space is isostatistic to the Stokes parameter space $(-1,1)^2$, 
		and the grayish disk corresponds to (a slice of) the Bloch ball 
		representing the quantum state space $\S(\C^2)$. 
		In the space $\P(\O^2)$ of probability distributions, 
		the MLE $p^*$ was the point in $\B$
		that is ``closest'' from the empirical distribution $\hat q_n$ 
		as measured by the Kullback-Leibler divergence.  
		Likewise, in the Stokes parameter space, 
		%this situation is interpreted as follows: 
		the MLE $\x^*$  that satisfies $p^*=p_{\x^*}$ 
		is the $\tilde\nabla^{(m)}$-projection 
		from the temporal estimate $\hat\x$ to the Bloch ball $B$. 
	\label{fig:indMfdTop}}
\end{figure}%---------------------------------------------------------------------------------

Incidentally, it should be noted that the isostatistic correspondence between $(-1,1)^k$ and $\P(\O)^{\otimes k}$ can be visualized by ``looking at $\P(\O)^{\otimes k}$ from the top.''
For instance, when $k=2$, the space $\P(\O)^{\otimes 2}$ was the ruled surface depicted in Fig.~\ref{fig:indMfd}. 
If we look at the space from the top (see Fig.~\ref{fig:indMfdTop}), we can find (a two-dimensional slice of) the Bloch ball embedded in the Stokes parameter space. 
This is because the diffeomorphism $f:(-1,1)^k\to \P(\O)^{\otimes k}$ is given by the affine transformation  (\ref{eqn:eta}).
% which relates $(\x_1,\dots,\x_k)$ with $(\y_1,\dots,\y_k)$. 
Recall that, in the proof of Lemma~\ref{lem:divergenceEquality}, we introduced a $\nabla^{(m)}$-affine coordinate system $\y=(\y_1,\dots,\y_k;\, \y_{k+1},\dots, \y_d)$ of $\P(\O^k)$,
% as (\ref{eqn:lem3eta}), 
the first $k$ components of which gave a $\overline\nabla^{(m)}$-affine coordinate system of $\P(\O)^{\otimes k}$.
If we look at the space $\P(\O^k)$ from a certain angle in such a way that the remaining $(d-k)$-components $(\y_{k+1},\dots, \y_d)$ are ``squashed,'' then we can visualize the shape of $\P(\O)^{\otimes k}$, which is affinely isomorphic to $(-1,1)^k$. 
This is the underlying mechanism behind Fig.~\ref{fig:indMfdTop}. 

% we note that the pair of coordinate systems $\th$ and $\y$ introduced in the proof of Lemma \ref{lem:divergenceEquality} generate a mutually orthogonal dual foliation, which provides a mixed coordinate system $(\y_1,\dots,\y_k;\,\th^{k+1},\dots,\th^d)$ of $\P(\O^k)$ \cite{AmariNagaoka}. 
%This gives an alternative view of Lemma \ref{lem:isostatistic} and Corollary \ref{cor:StokesMLE}. 
%Let $k=2$ and let us view Fig.~\ref{fig:indMfd} from the top. 
%Then everything looks like $(-1,1)^2$, see Fig.~\ref{fig:indMfdTop}.

%-------------------------------------------------------------------------------------------
\subsection{Computation of MLE}
%-------------------------------------------------------------------------------------------

%Now we have a tool of computing the MLE from the empirical distribution in a language of the Stokes parameter space. 
Let $\hat\x$ be the temporal estimate defined by (\ref{eqn:checkXi}), i.e., 
\[
 \hat\x=(\hat\x_1, \hat\x_2,\hat\x_3)
 :=\left(\frac{n_{1}^{+}-n_{1}^{-}}{N},\;\frac{n_{2}^{+}-n_{2}^{-}}{N},\;\frac{n_{3}^{+}-n_{3}^{-}}{N}\right).
\]
By a slight abuse of terminology, we shall call $\hat\x$ the {\em empirical distribution} on the Stokes parameter space. 
Suppose that the empirical distribution $\hat\x$ has fallen outside the Bloch ball $B$. 
Let $\x^*=(\x_1^*, \x^*_2,\x^*_3)$ be a point on the Bloch sphere $S$ in the Stokes parameter space. 
%According to the discussion up to here, 
If $\x^*$ is the MLE, then we see from Corollary~\ref{cor:StokesMLE} that 
the $\tilde\nabla^{(m)}$-geodesic (i.e., the straight line) connecting $\x^*$ and $\hat\x$ must be orthogonal to the Bloch sphere $S$ at $\x^*$ 
with respect to the induced Riemannian metric $\tilde g$.
Stated otherwise, the tangent vector $V$ of that geodesic at $\x^*$, which is explicitly given by
\begin{equation}\label{eqn:tangentGeodesic}
 V:=\sum_{i=1}^3 (\hat\x_i-\x^*_i)\left(\frac{\partial}{\partial\x_i}\right)_{\x^*},
\end{equation}
satisfies the orthogonality
\begin{equation}\label{eqn:orthogonalMLE}
 \tilde g_{\x^*}(V,X)=0
\end{equation}
for all tangent vectors $X\in T_{\x^*} S$ of the Bloch sphere $S$ at $\x^*$. 
%where $V$ is the tangent vector of the above-mentioned $\tilde\nabla^{(m)}$-geodesic at $\x^*$, which is explicitly given by
%\begin{equation}\label{eqn:tangentGeodesic}
% V:=\sum_{i=1}^3 (\hat\x_i-\x^*_i)\left(\frac{\partial}{\partial\x_i}\right)_{\x^*}.
%\end{equation}
The MLE $\x^*$ can be obtained as a solution of the equation (\ref{eqn:orthogonalMLE}). 
%using a suitably chosen basis $\{X_1, X_2\}$ of $T_{\x^*} S$ for $X$. 

Here we propose a method of computing the MLE $\x^*$.
%solving (\ref{eqn:orthogonalMLE}). 
%, by using ``normal'' vectors of the Bloch surface $S$. 
In Euclidean geometry, the position vector $\overrightarrow{\x}=(\x_1,\x_2, \x_3)$ of a point $\x$ on the unit sphere $S$ is normal to $S$, 
%orthogonal to every tangent vector $\overrightarrow{X}=(X_1,X_2,X_3)\in T_p S$ of $S$, 
in that they satisfy
\begin{equation}\label{eqn:normalityEuclid}
 \sum_{i=1}^3 \x_i X_i =0
\end{equation}
for all tangent vectors $\overrightarrow{X}=(X_1,X_2,X_3)\in T_\x S$ of $S$.
%with respect to the Euclidean metric.
%In other words, $\overrightarrow{p}$ is a normal vector of the unit sphere $S$ at $p\in S$ in Euclidean geometry. 
Using the relation (\ref{eqn:normalityEuclid}), we can find a tangent vector $\overrightarrow{n}$ at $\x\in S$ that is normal to $S$ with respect to the metric $\tilde g$. 
Let 
%$\mathbf{n}$ $\vec{n}$ 
\[
\overrightarrow{n}=\sum_{i=1}^3 a_i \left(\frac{\partial}{\partial\x_i}\right)_{\x}
\]
and let us represent a tangent vector $\overrightarrow{X}\in T_\x S$ of $S$ as
\[
 \overrightarrow{X}=\sum_{i=1}^3 X_i \left(\frac{\partial}{\partial\x_i}\right)_{\x}.
\]
The orthogonality with respect to $\tilde g$ is then written as
\begin{eqnarray*}
 \tilde g_\x(\overrightarrow{n}, \overrightarrow{X})
 &=&
 \sum_{i,j=1}^3 a_i X_j \, \tilde g_\x\left(\frac{\partial}{\partial\x_i}, \frac{\partial}{\partial\x_j}\right) \\
 &=&
 \sum_{i,j=1}^3 a_i X_j \, \frac{\d_{ij}}{1-(\x_i)^2} \\
 &=&
 \sum_{i=1}^3 \frac{a_i}{1-(\x_i)^2} \, X_i \\ \\
 &=&0.
\end{eqnarray*}
In the second equality, we used the explicit formula (\ref{eqn:FisherStokes}) for the Riemannian metric $\tilde g$.
Comparing this relation with (\ref{eqn:normalityEuclid}), we see that the choice
\[
 a_i:=\x_i \left(1-(\x_i)^2\right)
\]
gives a desired tangent vector $\overrightarrow{n}$ that is normal to $S$ at $\x$ with respect to $\tilde g$.

The condition (\ref{eqn:orthogonalMLE}) for the MLE $\x^*$ is now restated that the tangent vector (\ref{eqn:tangentGeodesic}) of the $\tilde\nabla^{(m)}$-geodesic should be parallel to the normal vector $\overrightarrow{n}$ at $\x^*$, so that there is a positive real number $\l$ such that
$\overrightarrow{n}=\l V$, or equivalently, 
\[
 \x^*_i \left(1-(\x^*_i)^2\right)=\l (\hat\x_i-\x^*_i),\qquad (i\in\{1,2,3\}).
\]
The MLE $\x^*$ is obtained by the unique solution of these equations together with the normalizing condition
\[
  \sum_{i=1}^3(\x^*_i)^2=1,
\]
and the positivity condition $\l>0$. 
The proof of Theorem~\ref{thm:main1} is now complete.

%----------------------------------------------------------------------------------------------------------------------------------
\section{Proof of Theorem \ref{thm:main2}}\label{sec:3}
%----------------------------------------------------------------------------------------------------------------------------------

Generalizing Theorem~\ref{thm:main1} to Theorem~\ref{thm:main2} is, in a sense, straightforward: 
we need only change the metric $\tilde g$ on $(-1,1)^3$ from (\ref{eqn:metric1}) to (\ref{eqn:metric2}) 
in the proof of Lemma~\ref{lem:isostatistic}, 
based on the fact that the Fisher information of i.i.d.~extensions of a statistical model increases linearly in the degree of extensions. 
%the isostatistic relationship between the Stokes parameter space $(-1,1)^3$ and the manifold of product distributions established in 
%Lemma~\ref{lem:isostatistic}, which results in changing the induced metric $\tilde g$ on $(-1,1)^3$ from (\ref{eqn:metric1}) to (\ref{eqn:metric2}). 
However, we here give an alternative proof, in order to reveal a different aspect of the quantum state tomography.
%, a randomized tomography. 

Let us consider the following experiment:  
One of the three observables $\s_1, \s_2$, and $\s_3$ is chosen at random with probability $s_1, s_2$, and $(1-s_1-s_2)$, respectively, 
and measure the chosen observable to yield an outcome either $+1$ or $-1$. 
We could estimate the unknown state $\r\in\S(\C^2)$ by repeating this randomized experiment. 
In particular, if $s_1=s_2=\frac{1}{3}$, this experiment is asymptotically equivalent to the standard quantum state tomography because of the law of large numbers. 
We shall call such an experiment a {\em randomized tomography} \cite{Yamagata:2011}. 

The sample space $\O$ for a randomized tomography is
\[
 \O=\{(\s_1,+1), (\s_1,-1),(\s_2,+1),(\s_2,-1),(\s_3,+1),(\s_3,-1)\}. 
\]
If the unknown state is specified by the Stokes parameters $\x=(\x_1,\x_2,\x_3)$, then the corresponding probability distribution on $\O$ is given by the probability vector
\[
 p_{(s,\x)}:=\left(s_1\frac{1+\x_1}{2}, s_1\frac{1-\x_1}{2}, s_2\frac{1+\x_2}{2}, s_2\frac{1-\x_1}{2}, 
 (1-s_1-s_2)\frac{1+\x_3}{2}, (1-s_1-s_2)\frac{1-\x_3}{2}\right),
\]
where $s:=(s_1,s_2)$ with the domain
\[
  D:=\{(s_1,s_2)\;|\; s_1>0,\; s_2>0,\; 1-s_1-s_2>0\}. 
\]
% is the parameter such that $s_1, s_2$ and $1-s_1-s_2$ are all positive. 
Note that the family 
\[
  \{p_{(s,\x)}\;|\;  s\in D,\, \x\in (-1,1)^3\}
\]
is identical to the five-dimensional probability simplex $\P(\O)$, and the parameters $(s,\x)$ form a coordinate system of $\P(\O)$. 
Since we are interested in estimating only the Stokes parameters $\x=(\x_1,\x_2,\x_3)$, the remaining parameters $s=(s_1,s_2)$ are regarded as nuisance parameters \cite{{LehmanCasella}, {AmariLN}} in the terminology of statistics. 
In what follows, $\P(\O)$ is regarded as a statistical manifold endowed with the dualistic structure 
$(g,\nabla^{(e)},\nabla^{(m)})$, where $g$ is the Fisher metric, and $\nabla^{(e)}$ and $\nabla^{(m)}$ are 
the exponential and mixture connections.

Let us consider the following submanifolds of $\P(\O)$: 
\[
 M(s):=\{p_{(s,\x)}\;|\;  \x\in (-1,1)^3\}
\]
for each $s\in D$, and
\[
 E(\x):=\{p_{(s,\x)}\;|\; s\in D\}
\]
for each $\x\in (-1,1)^3$.
Since $M(s)$ and $E(\x)$ are convex subsets of $\P(\O)$, they are $\nabla^{(m)}$-autoparallel. 
The following Lemma is the key to the estimation of $\x$ under the nuisance parameters $s$. 

\begin{lemma}\label{lem:foliation}
For each $\x\in(-1,1)^3$, the submanifold $E(\x)$ is $\nabla^{(e)}$-autoparallel. 
Furthermore, for each $s\in D$ and $\x\in(-1,1)^3$, 
the submanifolds $M(s)$ and $E(\x)$ are mutually orthogonal with respect to the Fisher metric $g$. 
\end{lemma}

\begin{proof}
Let us change the coordinate system $(s,\x)=(s_1,s_2, \x_1,\x_2,\x_3)$ into
\[
 \y:=(\y_1,\,\y_2,\,\y_3,\,\y_4,\,\y_5):=(s_1,\, s_2,\, s_1 \x_1,\, s_2\x_2,\, (1-s_1-s_2)\x_3).
\]
With this coordinate transformation, the probability vector $p_{(s,\x)}$ is rewritten as
\[
 p_\y:=\left(\frac{\y_1+\y_3}{2}, \frac{\y_1-\y_3}{2}, \frac{\y_2+\y_4}{2}, \frac{\y_2-\y_4}{2}, 
 \frac{1-\y_1-\y_2+\y_5}{2}, \frac{1-\y_1-\y_2-\y_5}{2}\right).
\]
We see from this expression that the coordinate system $\y:=(\y_i)_{1\le i\le 5}$ is $\nabla^{(m)}$-affine. 
The potential function for $\y$ is given by the negative entropy
\[ \varphi(\y):=\sum_{\o\in\O} p_\y(\o) \log p_\y(\o), \]
and the dual $\nabla^{(e)}$-affine coordinate system $\th=(\th^i)_{1\le i\le 5}$ is given by
\[
 \th^i:=\frac{\partial \varphi}{\partial \y_i}. 
\]
By direct computation, we have
\begin{eqnarray*}
&&\th^1=\frac{1}{2}\log\frac{(\y_1+\y_3)(\y_1-\y_3)}{(1-\y_1-\y_2+\y_5)(1-\y_1-\y_2-\y_5)}
	=\frac{1}{2}\log\left[\left(\frac{s_1}{1-s_1-s_2}\right)^2 \frac{1-(\x_1)^2}{1-(\x_3)^2}\right], \\ \\
&&\th^2=\frac{1}{2}\log\frac{(\y_2+\y_4)(\y_2-\y_4)}{(1-\y_1-\y_2+\y_5)(1-\y_1-\y_2-\y_5)}
	=\frac{1}{2}\log\left[\left(\frac{s_2}{1-s_1-s_2}\right)^2 \frac{1-(\x_2)^2}{1-(\x_3)^2}\right], \\ \\
&&\th^3=\frac{1}{2}\log\frac{\y_1+\y_3}{\y_1-\y_3}=\frac{1}{2}\log\frac{1+\x_1}{1-\x_1}, \\ \\
&&\th^4=\frac{1}{2}\log\frac{\y_2+\y_4}{\y_2-\y_4}=\frac{1}{2}\log\frac{1+\x_2}{1-\x_2}, \\ \\
&&\th^5=\frac{1}{2}\log\frac{1-\y_1-\y_2+\y_5}{1-\y_1-\y_2-\y_5}=\frac{1}{2}\log\frac{1+\x_3}{1-\x_3}. 
\end{eqnarray*}
Thus, fixing $\x$ is equivalent to fixing the three coordinates $(\th^3,\th^4,\th^5)$, 
and the submanifold $E(\x)$ is generated by changing the remaining two parameters $(\th^1, \th^2)$. 
This implies that $E(\x)$ is $\nabla^{(e)}$-autoparallel, proving the first part of the claim. 

To prove the second part, let us introduce a mixed coordinate system \cite{AmariNagaoka}
\[  (\y_1,\y_2;\,\th^3,\th^4,\th^5) \] 
of $\S(\O)$. 
Since $(\y_1, \y_2)=(s_1,s_2)$, the submanifold $M(s)$ is rewritten as
\[
 M(s)=\{p_{(s,\x)}\;|\; \mbox{$(\y_1,\y_2)$ are fixed and $(\th^3,\th^4,\th^5)$ are arbitrary}\}.
\]
On the other hand, as was seen in the above, the submanifold $E(\x)$ is rewritten as
\[
 E(\x)=\{p_{(s,\x)}\;|\; \mbox{$(\th^3,\th^4,\th^5)$ are fixed and $(\y_1,\y_2)$ are arbitrary} \}.
\]
Thus the general orthogonality relation
\[
 g\left(\frac{\partial}{\partial\th^i}, \frac{\partial}{\partial\y_j}\right)=\d_i^j
\]
proves that $M(s)$ and $E(\x)$ are orthogonal to each other. 
\end{proof}

\begin{figure}[t] %---------------------------------------------------------------------------------
	\begin{centering}
	\includegraphics[scale=0.5]{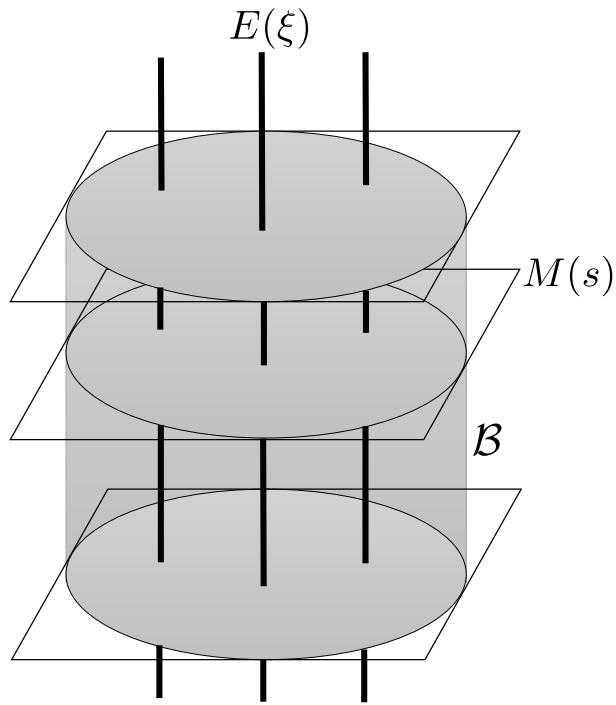}
	\par
	\end{centering}
	\caption{Mutually orthogonal dualistic foliation of $\P(\O)$ 
	based on $M(s)$ and $E(\x)$. 
	Each section $M(s)$ is affinely isomorphic to the Stokes parameter space $(-1,1)^3$. 
	The grayish cylindrical area indicates the subset $\B=\{p_{(s,\x)}| s\in D,\,\x\in B\}$ 
	of $\P(\O)$ that corresponds to the Bloch ball $B$. 
	In particular, for each $s\in D$, the intersection $M(s) \cap \B$ is affinely isomorphic 
	to the Bloch ball $B$.
	%The circle depicted on each $M(s)$ indicates the Bloch ball. 
	\label{fig:foliation}}
\end{figure}%---------------------------------------------------------------------------------

Lemma~\ref{lem:foliation} implies that the manifold $\P(\O)$ is decomposed into a mutually orthogonal dualistic foliation based on the submanifolds $M(s)$ and $E(\x)$, as illustrated in Fig.~\ref{fig:foliation}. 

Let us get down to the problem of estimating the unknown Stokes parameters $\x$ using the randomized tomography. 
Suppose that, among $N$ independent experiments of randomized tomography, 
the $i$th Pauli matrix $\s_i$ was measured $N_i$ times and obtained outcomes $+1$ and $-1$, each $n_i^+$ and $n_i^-$ times. 
Then a temporal estimate $(\hat s,\hat\x)$ for the parameters $(s,\x)$ are
\[
 \hat s=\left(\frac{N_1}{N}, \frac{N_2}{N}\right)
\]
and
\[
 \hat\x
 =\left(\frac{n_1^+ -n_1^-}{N_1}, \frac{n_2^+ -n_2^-}{N_2}, \frac{n_3^+ -n_3^-}{N_3} \right).
\]
If $\hat\x$ has fallen outside the Bloch ball $B$, we may find a corrected estimate by the maximum likelihood method. 
First of all, the empirical distribution $\hat q_N\in\S(\O)$ is given by
\[
 \hat q_N:=p_{(\hat s,\hat\x)}. 
\]
On the other hand, the Bloch ball $B$ in the Stokes parameter space $(-1,1)^3$ corresponds to the subset
\[
 \B:=\{p_{(s,\x)}\;|\; s\in D,\;\x\in B\}
\]
of $\P(\O)$, (cf., Fig.~\ref{fig:foliation}).
%, which forms a five-dimensional tubular region.  
The MLE $p^*$ in $\P(\O)$ is then given by 
\begin{equation}\label{eqn:foliationMLE}
 p^*=\argmin_{p\in\B} D(\hat q_N\| p).
\end{equation}
This is the $\nabla^{(m)}$-projection from the empirical distribution $\hat q_N$ to $\B$.
A crucial observation is the following

\begin{lemma}\label{lem:foliation2}
The minimum in (\ref{eqn:foliationMLE}) is achieved on $M(\hat s)\cap\B$.  
\end{lemma}

\begin{proof}
Let us take a point $p_{(s,\x)}\in\B$ arbitrarily. 
It then follows from the mutually orthogonal dualistic foliation of $\P(\O)$ established in Lemma~\ref{lem:foliation} that 
\begin{eqnarray*}
 D(\hat q_N\| p_{(s,\x)})
 &=&D(p_{(\hat s,\hat\x)}\| p_{(s,\x)}) \\
 &=&D(p_{(\hat s,\hat\x)}\| p_{(\hat s,\x)})+D(p_{(\hat s,\x)}\| p_{(s,\x)}) \\
 &\ge&D(p_{(\hat s,\hat\x)}\| p_{(\hat s,\x)}).
\end{eqnarray*}
In the second equality, the generalized Pythagorean theorem was used. 
Consequently, 
\[
 \min_{\x\in B} D(p_{(\hat s,\hat\x)}\| p_{(s,\x)}) \ge \min_{\x\in B} D(p_{(\hat s,\hat\x)}\| p_{(\hat s,\x)})
\]
for all $s\in D$, and the lower bound is achieved if and only if $s=\hat s$.
\end{proof}

\begin{figure}[t] %---------------------------------------------------------------------------------
	\begin{centering}
	\includegraphics[scale=0.5]{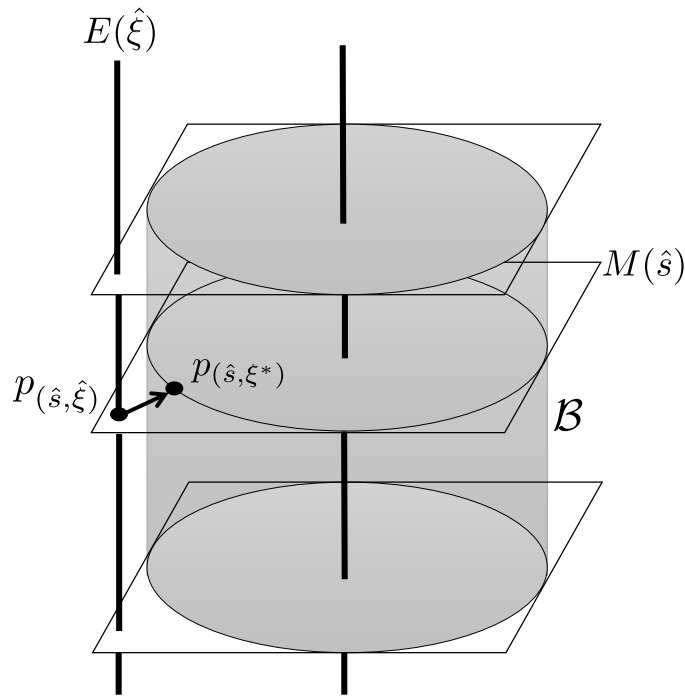}
	\par
	\end{centering}
	\caption{The maximum likelihood method in the framework of randomized tomography.
	%is understood from the viewpoint of orthogonal dual foliation. 
	%The section $M(\hat s)$ corresponding to $s=\hat s$ gives the most suitable 
	%geometrical structure for estimating $\x$. 
	Given a temporal estimate $(\hat s, \hat\x)$ with $\hat\x\notin B$, 
	we can restrict ourselves to the slice $M(\hat s)$ as the search space for the MLE $p^*$, 
	and $p^*=p_{(\hat s, \x^*)}$ %that corresponds to the MLE $\x^*$ 
	is the $\nabla^{(m)}$-projection from the empirical distribution $p_{(\hat s,\hat\x)}$ 
	to $\B$ on the slice $M(\hat s)$. 
	\label{fig:foliation2}}
\end{figure}%---------------------------------------------------------------------------------

The geometrical implication of Lemma~\ref{lem:foliation2} is illustrated in Fig.~\ref{fig:foliation2}. 
The MLE $p^*=p_{(\hat s, \x^*)}$ is the $\nabla^{(m)}$-projection from the empirical distribution $p_{(\hat s,\hat\x)}$ to $\B$ on the slice $M(\hat s)$.  

Now we are ready to prove Theorem~\ref{thm:main2}.
Suppose we are given a temporal estimate $(\hat s, \hat\x)$ with $\hat\x\notin B$.
% has fallen outside the Bloch ball $B$. 
Due to Lemma~\ref{lem:foliation2}, we can restrict ourselves to the slice $M(\hat s)$ as the search space for the MLE $p^*$. 
Since the slice $M(\hat s)$ is affinely isomorphic to the Stokes parameter space $(-1,1)^3$, we can introduce a dualistic structure $(\tilde g, \tilde\nabla^{(e)}, \tilde\nabla^{(m)})$ on $(-1,1)^3$ in the following way. 
Firstly, in a quite similar way to the derivation of (\ref{eqn:FisherStokes}), it is shown that 
the components of the Fisher metric $g$ on the section $M(\hat s)$ with respect to the coordinate system $\x=(\x_1,\x_2,\x_3)$ are given by
\[ %\begin{equation}\label{eqn:metricOnM}
  g_{(\hat s,\x)}\left(\frac{\partial}{\partial\x_i}, \frac{\partial}{\partial\x_j}\right)
  =\frac{\hat s_i\, \d_{ij}}{1-(\x_i)^2},
\] %\end{equation}
where $\hat s_3:=1-\hat s_1-\hat s_2$. 
We identify this metric with $\tilde g$, i.e., 
\[
 \tilde g_{\x}\left(\frac{\partial}{\partial\x_i}, \frac{\partial}{\partial\x_j}\right)
 :=\frac{\hat s_i\, \d_{ij}}{1-(\x_i)^2}.
\]
Secondly, the mixture connection $\tilde\nabla^{(m)}$ is defined so that the coordinate system $\x=(\x_1,\x_2,\x_3)$ of $(-1,1)^3$ becomes $\tilde\nabla^{(m)}$-affine.
Finally, the dual connection $\tilde\nabla^{(e)}$ is defined by the duality
\[
 \tilde g(\tilde\nabla^{(e)}_XY,Z):=X\tilde g(Y,Z)-\tilde g(Y, \tilde\nabla^{(e)}_XZ).
\]
It is shown, in a quite similar way to the proof of Lemma~\ref{lem:isostatistic}, that the statistical manifold $((-1,1)^3, \tilde g, \tilde\nabla^{(e)}, \tilde\nabla^{(m)})$ is isostatistic to the manifold $M(s)$ with a dualistic structure defined by the restriction of $(g, \nabla^{(e)}, \nabla^{(m)})$ to $M(s)$. 
Thus, the MLE $\x^*$ in the Stokes parameter space is given by the $\tilde\nabla^{(m)}$-projection from $\hat\x$ to the Bloch sphere $S$ with respect to the metric $\tilde g$. 
This proves the first part of Theorem~\ref{thm:main2}.
The remainder of Theorem~\ref{thm:main2} is proved in the same way as the corresponding part of Theorem~\ref{thm:main1}. 

\begin{figure}[t] %---------------------------------------------------------------------------------
	\begin{centering}
	\includegraphics[scale=0.8]{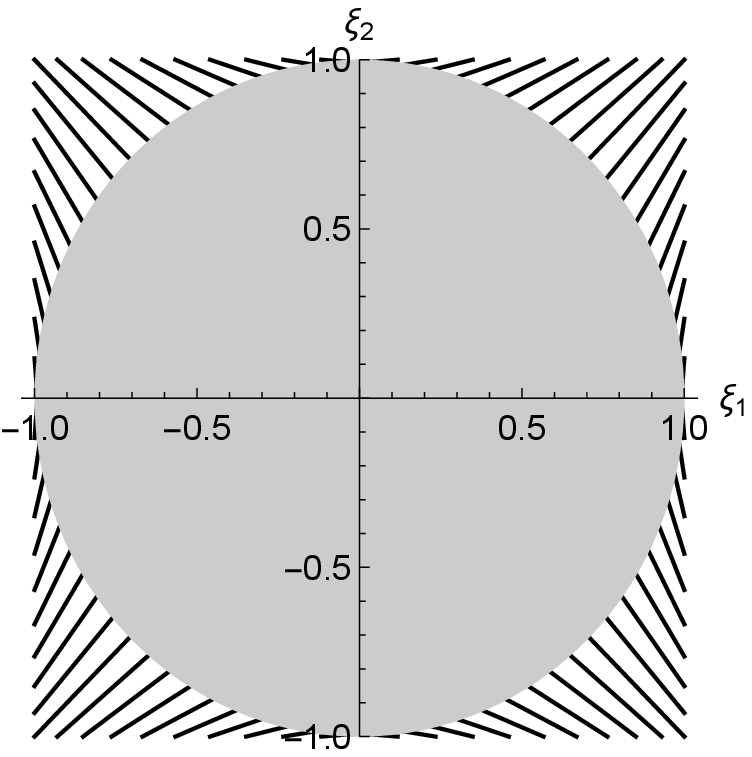}$\quad$
	\includegraphics[scale=0.8]{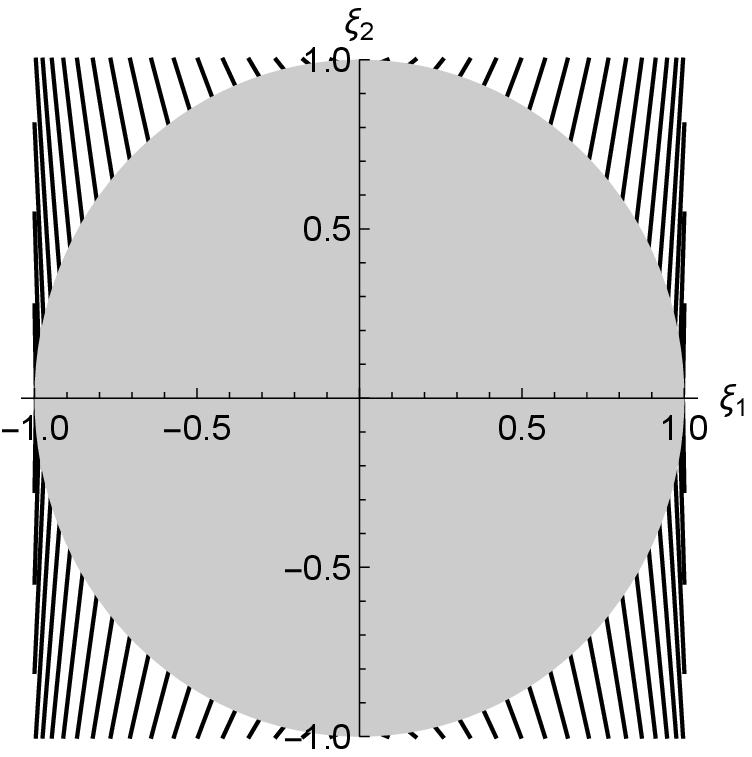}
	\par
	\end{centering}
	\caption{The trajectories of $\tilde\nabla^{(m)}$-projections on the $\x_1\x_2$-plane 
	that give the MLE $p^*$ 
	when $N_1:N_2=1:1$ (left), and $N_1:N_2=5:1$ (right).
	The change of $\x_1$-coordinate relative to the change of $\x_2$-coordinate 
	along each trajectory is less noticeable in the right panel than in the left panel. 
	%The ratio of the change of $\x_1$-coordinate to that of $\x_2$-coordinate 
	%along each trajectory in the right panel is smaller than that in the left panel. 
	This is because a tomography with $N_1/N_2=5$ provides us with more information 
	about $\x_1$, relative to $\x_2$, as compared to that with $N_1/N_2=1$. 
	\label{fig:projection}}
\end{figure}%---------------------------------------------------------------------------------

Fig.~\ref{fig:projection} demonstrates how the $\tilde\nabla^{(m)}$-projection is realized on the $\x_1\x_2$-plane of the Stokes parameter space: 
the left and right panels correspond to the cases when 
$N_1:N_2=1:1$ and $N_1:N_2=5:1$, respectively.
The change of $\x_1$-coordinate relative to the change of $\x_2$-coordinate along each trajectory is less noticeable in the right panel than in the left panel. 
This is because a tomography with $N_1/N_2=5$ provides us with more information about $\x_1$-coordinate, relative to $\x_2$-coordinate, as compared to that with $N_1/N_2=1$.

%----------------------------------------------------------------------------------------------------------
\section{Numerical demonstration}\label{sec:4}
%----------------------------------------------------------------------------------------------------------

In this section, we devise a method of computing the MLE based on Theorem~\ref{thm:main2}. 
%Suppose that, among $N$ independent experiments, 
%the $i$th Pauli matrix $\s_i$ is measured $N_i$ times and obtained outcomes $+1$ and $-1$, each $n_i^+$ and $n_i^-$ times. 
%Then the temporal estimate for the unknown Stokes parameters $\x$ is
Suppose we are given a temporal estimate
\[
 \hat\x=(\hat\x_1, \hat\x_2,\hat\x_3)
 :=\left(\frac{n_{1}^{+}-n_{1}^{-}}{N_1},\;\frac{n_{2}^{+}-n_{2}^{-}}{N_2},\;\frac{n_{3}^{+}-n_{3}^{-}}{N_3}\right).
\]
%\[ (\hat s_1,\hat s_2,\hat s_3) =\left(\frac{N_1}{N}, \frac{N_2}{N}, \frac{N_3}{N}\right) \]
If $\|\hat\x \|\le 1$, then $\hat\x$ already gives a valid estimate (in fact the MLE) for $\x$. 
%If $\|\hat\x \| > 1$, on the other hand, then 
Otherwise, the estimate is corrected using the method stated in Theorem \ref{thm:main2}: 
the corrected estimate $\x^*=(\x^*_1, \x^*_2, \x^*_3)$ is the unique solution of the simultaneous equations
\begin{equation}\label{eqn:cubic}
\x^*_i \left(1-(\x^*_i)^2\right)=\l \hat s_i\,(\hat\x_i-\x^*_i),\qquad (i\in\{1,2,3\})
\end{equation}
and
\begin{equation}\label{eqn:quadratic}
  (\x^*_1)^2+(\x^*_2)^2+(\x^*_3)^2=1, 
\end{equation}
with $\l>0$. 

Let us consider, for each $a\in (-1,1)$, the following cubic equation in $x$:
\[
 x(1-x^2)=\m (a-x).
\]
This equation has a unique solution 
\begin{equation}\label{eqn:sol}
 x=({\rm sgn}\,a) \frac{2\sqrt{\m+1}}{\sqrt{3}}\cos\left[
 	\frac{1}{3}\left(\pi+
	\arctan\sqrt{\frac{4(\m+1)^{3}}{27\m^{2}a^{2}}-1}
	\right)\right]
\end{equation}
in the interval $-1< x <1$. 
Let us denote the right-hand side of (\ref{eqn:sol}) as $x(\m,a)$. 
Then the solution of each equation in (\ref{eqn:cubic}) is given by $\x^*_i=x(\l \hat s_i, \hat\x_i)$, and the norm condition (\ref{eqn:quadratic}) is reduced to
\begin{equation}\label{eqn:quadratic2}
 x(\l \hat s_1, \hat\x_1)^2+x(\l \hat s_2, \hat\x_2)^2+x(\l \hat s_3, \hat\x_3)^2=1. 
\end{equation}
This is an equation for a single variable $\l$. 
Let $\l^*$ be the unique positive solution of (\ref{eqn:quadratic2}). 
Then the MLE is given by
\[
 \x^*_i=x(\l^* \hat s_i, \hat\x_i),\qquad (i\in\{1,2,3\}).
\]

In practice, the solution $\l^*$ cannot be obtained explicitly: 
%analytically in an explicit form. 
thus, we must invoke numerical evaluation.
% by using, for example, the Newton method. 
%A possible initial value for the Newton method might be $\l=1$.
%We made numerical demonstration to verify if our method is efficient. 
%For the sake of numerical demonstration, 
For the sake of demonstration,
%demonstrating the efficiency of our method, 
we computed the MLE 1000 times on MATHEMATICA software version 10.4, 
using (i) FindRoot function to solve  (\ref{eqn:quadratic2}),
%with $\hat{s}_1=\hat{s}_2=\hat{s}_3=\frac{1}{3}$, 
and (ii) FindMaximun function to find the maximizer (\ref{eqn:naiveMLE}) directly,
under the condition that $N_1=N_2=N_3$, 
starting from randomly generated initial points $(\hat\x_1,\hat\x_2,\hat\x_3)$ that fall outside the Bloch ball. 
%uniformly in the region $(\frac{1}{\sqrt{3}},1)^3$),
The average computation time was 2.20313 [msec] for (i), and 21.6406 [msec] for (ii). 
As far as this demonstration is concerned, our method works very efficiently. 

We note that the present method has been successfully applied to 
%usefulness of the present method has been demonstrated also in 
an experimental study using photonic qubits \cite{OkamotoOYFT:2016}.

%----------------------------------------------------------------------------------------------------------
\section{Conclusions}\label{sec:5}
%----------------------------------------------------------------------------------------------------------

In the present paper, a statistically feasible method of data post-processing for the quantum state tomography was studied from an information geometrical point of view. 
Suppose that, among $N$ independent experiment, the $i$th Pauli matrix $\s_i$ was measured $N_i$ times and obtained outcomes $+1$ and $-1$, each $n^+_i$ and $n^-_i$ times. 
Then the space $(-1,1)^3$ of the Stokes parameter $\x=(\x_1,\x_2,\x_3)$ should be regarded as a Riemannian manifold endowed with a metric 
\[
g_\x\left(\frac{\partial}{\partial\x_i}, \frac{\partial}{\partial\x_j}\right)
  =\frac{\hat s_i\, \d_{ij}}{1-(\x_i)^2},
\]
where $\hat s_i:=N_i/N$.
Furthermore, if the temporal estimate
\[
 \hat\x=\left(\frac{n^+_1-n^-_1}{N_1},\frac{n^+_2-n^-_2}{N_2}, \frac{n^+_3-n^-_3}{N_3}\right) 
\]
for the parameter $\x$ has fallen outside the Bloch ball, then the maximum likelihood estimate (MLE) is the orthogonal projection from $\hat\x$ onto the Bloch sphere with respect to the metric $g$ defined above. 
An efficient algorithm for finding the MLE was also proposed. 

%Extending this result to more general types of quantum state tomography is an interesting future project.

%----------------------------------------------------------------------------------------------------------------------------------
\section*{Acknowledgment}
%----------------------------------------------------------------------------------------------------------------------------------

The authors are grateful to Professors Ryo Okamoto and Shigeki Takeuchi for helpful discussions. 
The present study was supported by JSPS KAKENHI Grant Number JP22340019.

%----------------------------------------------------------------------------------------------------------------------------------
%----------------------------------------------------------------------------------------------------------------------------------
%----------------------------------------------------------------------------------------------------------------------------------
%----------------------------------------------------------------------------------------------------------------------------------
%----------------------------------------------------------------------------------------------------------------------------------
\appendix
\section*{Appendix: Information geometry: an overview} %\label{app:infoGeo}
%----------------------------------------------------------------------------------------------------------------------------------
%\section{Information geometry: an overview}\label{app:infoGeo}
%----------------------------------------------------------------------------------------------------------------------------------

In this appendix, we give a brief summary of information geometry. 
Suppose we are given a Riemannian manifold $(M,g)$, where $M$ is an $n$-dimensional differentiable manifold and $g$ is a metric. 
A pair of affine connections, $\nabla$ and $\nabla^{*}$, on $(M,g)$ 
%represented by covariant derivatives $\nabla$ and $\nabla^{*}$ 
are said to be {\em mutually dual} with respect to $g$ if they satisfy
\begin{equation}\label{eqn:duality}
Xg(Y,Z)=g(\nabla_{X}Y,Z)+g(Y,\nabla_{X}^{*}Z)
\end{equation}
for vector fields $X,Y$, and $Z$ on $M$. 
A triad $(g,\nabla,\nabla^*)$ satisfying the duality (\ref{eqn:duality}) is called a {\em dualistic structure} on $M$. 
If Riemannian curvatures and torsions of $\nabla$ and $\nabla^{*}$ all vanish, then $M$ is said to be {\em dually flat}. 

For a dually flat manifold $(M,g,\nabla, \nabla^{*})$, we can construct a pair of affine coordinate systems in the following way. 
Since $M$ is $\nabla$-flat, there is a $\nabla$-affine coordinate system $\th=(\th^i)_{1\le i\le n}$. 
Likewise, since $M$ is $\nabla^*$-flat, there is a $\nabla^*$-affine coordinate system $\y=(\y_i)_{1\le i\le n}$. 
Furthermore, we can choose $\th$ and $\y$ in such a way that they satisfy the orthogonality:
\[ %\begin{equation}\label{eqn:delta_ij}
g\left(\frac{\partial}{\partial\theta^{i}},\frac{\partial}{\partial\eta_{j}}\right)=\delta_{i}^{j}.
\] %\end{equation}
Such a pair of $\nabla$- and $\nabla^*$-affine coordinate systems $\{\th, \y\}$ is said to be {\em mutually dual} with respect to the dualistic structure $(g,\nabla,\nabla^*)$. 
 
By using dual affine coordinate systems $\{\th, \y\}$, we can construct a pair of canonical divergences on a dually flat manifold $(M,g,\nabla, \nabla^{*})$ as follows. 
We first find a pair of potential functions $\psi(\th)$ and $\varphi(\y)$ on $M$ that satisfy
\[
 \th^i=\partial^i\varphi(\y),\qquad \y_i=\partial_i\psi(\th),\qquad
 \psi(\th)+\varphi(\y)-\sum_{i}\th^i\y_i=0,
\]
where $\partial^i:=\partial/\partial\eta_i$ and $\partial_i:=\partial/\partial\th^i$. 
By using these potentials, we define the {\em $\nabla$-divergence} $D^\nabla$ from $p\in M$ to $q\in M$ as
\[
 D^\nabla(p\|q):=\psi(\th(p))+\varphi(\y(q))-\sum_{i} \th^i(p)\y_i(q),
\]
where $\th(p)=(\th^i(p))_{1\le i\le n}$ and $\y(q)=(\y^i(q))_{1\le i\le n}$ are the $\th$-coordinate of $p$ and $\y$-coordinate of $q$, respectively. 
The other divergence $D^{\nabla^*}$, called the {\em $\nabla^*$-divergence}, is defined by changing the role of $\nabla$ and $\nabla^*$, to obtain
\[
 D^{\nabla^*}(p\| q):=D^{\nabla}(q\| p).
\]
It is shown that $D^{\nabla}(p \| q)\ge 0$ for all $p,q\in M$, and $D^{\nabla}(p \| q)=0$ if and only if $p=q$.

Incidentally, we note that the components of the metric $g$ with respect to the coordinate systems $\th$ and $\eta$ are given, respectively, by
\[
 g_{ij}:=g\left(\frac{\partial}{\partial\th^{i}},\frac{\partial}{\partial\th^{j}}\right)=\partial_i\partial_j\psi(\th),
\]
and
\[
 g^{ij}:=g\left(\frac{\partial}{\partial\y_{i}},\frac{\partial}{\partial\y_{j}}\right)=\partial^i\partial^j\varphi(\y). 
\]
The notations $g_{ij}$ and $g^{ij}$ fulfill the convention in tensor analysis, in that the inverse of the matrix $[g_{ij}]_{1\le i,j\le n}$ is actually identical to $[g^{ij}]_{1\le i,j\le n}$.

Now, let $M$ be a generic differentiable manifold endowed with an affine connection $\nabla$.
A submanifold $S$ of $M$ is called {\em $\nabla$-autoparallel} if $(\nabla_XY)_p\in T_pS$ for all vector fields $X$ and $Y$ on $S$, and all $p\in S$. 
In particular, a one-dimensional $\nabla$-autoparallel submanifold is called a {\em $\nabla$-geodesic}. 

Returning to a dually flat manifold $(M,g,\nabla, \nabla^{*})$, 
the $\nabla$-geodesic connecting two points $p$ and $q$ on $M$ is represented in terms of the $\nabla$-affine coordinate system $\theta$ as 
\[
\left\{ \left.\theta(p)+t(\theta(q)-\theta(p))\right|0\leq t\leq1\right\}.
\]
Similarly, the $\nabla^{*}$-geodesic connecting $p,q \in M$ is represented in terms of the $\nabla^*$-affine coordinate system $\y$ as 
\[
 \left\{ \left.\eta(p)+t(\eta(q)-\eta(p))\right|0\leq t\leq1\right\}.
\]
For three points $p,q$, and $r$ in $M$, we have
\[
D^{\nabla}(p \| q)+D^{\nabla}(q \| r)-D^{\nabla}(p \| r)=\sum_i(\th^i(p)-\th^i(q))(\y_i(r)-\y_i(q)).
\] 
It follows from this identity that, if $\nabla$-geodesic connecting $p,q$ and $\nabla^{*}$-geodesic connecting $q,r$
are orthogonal at $q$ with respect to the metric $g$, then the following {\em generalized Pythagorean theorem} holds (cf., Fig.~\ref{fig:Pythagoras}). 
\begin{equation}\label{eqn:Pythagoras}
D^{\nabla}(p \| q)+D^{\nabla}(q \| r)=D^{\nabla}(p \| r).
\end{equation}

\begin{figure}[t] %---------------------------------------------------------------------------------
	\begin{centering}
	\includegraphics[scale=0.5]{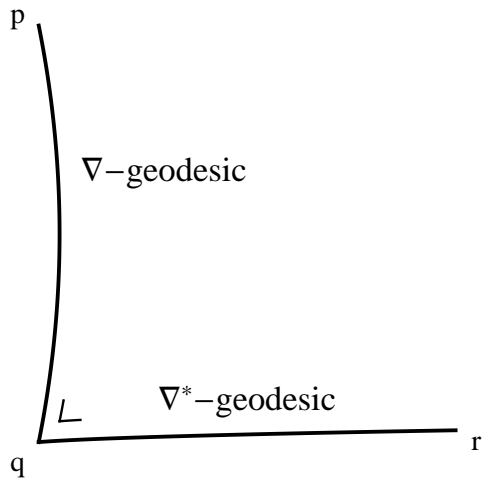}
	\par
	\end{centering}
	\caption{If $\nabla$-geodesic connecting $p,q$ and $\nabla^{*}$-geodesic connecting $q,r$
		are orthogonal at $q$ with respect to $g$, the generalized Pythagorean theorem 
		$D^{\nabla}(p \| q)+D^{\nabla}(q \| r)=D^{\nabla}(p \| r)$ holds. 
	\label{fig:Pythagoras}}
\end{figure}%---------------------------------------------------------------------------------

Given a (closed) submanifold $S$ of $M$ and a point $p\in M\backslash S$, let $q^*\in S$ be the point on $S$
%, if it exists, 
that is ``closest'' from $p$ as measured by the $\nabla$-divergence $D^{\nabla}$, i.e., 
\[
q^*:=\argmin_{r\in\S}D^{\nabla}(p \| r).
\]
Then, due to the generalized Pythagorean theorem (\ref{eqn:Pythagoras}), the point $q^*$ is 
%(either on the boundary of $S$ or) 
the $\nabla$-projection from $p$ to $S$ or its boundary.

A typical example of a dually flat manifold appears in statistics. 
The totality $\P(\O)$ of probability distributions on a finite sample space $\O$
is a $(|\O|-1)$-dimensional dually flat manifold with respect to the dualistic structure $(g,\nabla^{(e)},\nabla^{(m)})$,
where $g$ is the {\em Fisher metric}: 
\[
 g_p(X,Y):=\sum_{\o\in\O} p(\o)\left(X\log p(\o)\right)\left(Y\log p(\o)\right),
\]
$\nabla^{(e)}$ is the {\em exponential connection}:
\[
 g_p(\nabla^{(e)}_XY,Z):=\sum_{\o\in\O} \left(XY\log p(\o)\right)\left(Z p(\o)\right),
\]
and $\nabla^{(m)}$ is the {\em mixture connection}:
 \[
 g_p(\nabla^{(m)}_XY,Z):=\sum_{\o\in\O} \left(XY p(\o)\right)\left(Z \log p(\o)\right).
\]

Observe that each point $p\in\P(\O)$ is represented in the form
\[
 p(\o)=p_\y(\o):=\sum_{i=1}^{|\O|-1}\y_i\d_i(\o)+\left(1-\sum_{i=1}^{|\O|-1}\y_i\right)\d_n(\o),
 \qquad (\o\in\O),
\]
where $\d_i(\o)$ is the $\d$-measure concentrated on the $i$th outcome $\o_i$. 
Thus, the parameters $\y=(\y^i)_{1\le i\le |\O|-1}$ form a $\nabla^{(m)}$-affine coordinate system. 
The dual $\nabla^{(e)}$-affine coordinate system $\th=(\th^i)_{1\le i\le |\O|-1}$ is given by
\[
 \th^i=\log\frac{p(i)}{p(n)}.
\]
The potential functions $\psi(\th)$ and $\varphi(\y)$ are 
\[
 \psi(\th)=\log\left(1+\sum_{i=1}^{|\O|-1} e^{\th^i}\right). 
\]
and
\[
 \varphi(\y)=\sum_{i=1}^{|\O|-1} \y_i\log \y_i
 	+\left(1-\sum_{i=1}^{|\O|-1}\y_i\right)\log\left(1-\sum_{i=1}^{|\O|-1}\y_i\right).
\]
Note that $\varphi(\y)$ is the negative entropy of $p_\y$. 
By using these potential functions, a pair of divergence functions 
%$D^{\nabla^{(e)}}$ and $D^{\nabla^{(m)}}$ 
are defined. 
In particular, the $\nabla^{(m)}$-divergence $D^{\nabla^{(m)}}(p\| q)$ turns out to be identical to the {\em Kullback-Leibler divergence}
\[
 D(p\|q)=\sum_{\o\in\O} p(\o)\log\frac{p(\o)}{q(\o)}.
\]

A family $\{p_{\overline\th}(\o)\}_{\overline\th}$ of probability distributions parametrized by $\overline\th=(\th^1,\dots,\th^k)$ is called a $k$-dimensional {\em exponential family} if it takes the form
\[
p_{\overline\th}(\o)=\exp\left[ C(\o)+ \sum_{i=1}^k \th^i F_i(\o)-\psi(\overline\th) \right],
%\qquad\left(\overline\th=(\th^1,\dots,\th^k) \right)
\]
and a family $\{p_{\overline\y}(\o)\}_{\overline\y}$ of probability distributions parametrized by $\overline\y=(\y_1,\dots,\y_k)$ is called a $k$-dimensional {\em mixture family} if it takes the form
\[
p_{\overline\y}(\o)=\sum_{i=1}^k \y_i\, p_i(\o)+\left(1-\sum_{i=1}^k \y_i\right) p_0(\o).
%,\qquad\left(\overline\y=(\y_1,\dots,\y_k) \right).
\]
It is shown that a submanifold $S$ of $\P(\O)$ is $\nabla^{(e)}$-autoparallel if and only if it is an exponential family, 
and that $S$ is $\nabla^{(m)}$-autoparallel if and only if it is a mixture family. 
For more information, consult \cite{{AmariNagaoka},{AmariLN},{MurrayRice}}. 

%----------------------------------------------------------------------------------------------------------------------------------


\begin{thebibliography}{99}
\bibitem{note}
In the present paper, we use the term ``maximum likelihood estimate'' for both the parameter $\hat\th_n$ and the corresponding probability distribution $p_{\hat\th_n}$.
\end{thebibliography}
\end{document}